\documentclass[12pt,letterpaper]{article}
\usepackage{epsfig,rotating,setspace,latexsym,amsmath,epsf,amssymb,bm,theorem}
\usepackage{cite,appendix}

\usepackage{amsfonts}

\title{Secure Lossy Transmission of Vector Gaussian Sources\thanks{This work was supported by NSF Grants CCF 07-29127, CNS 09-64632, CCF
09-64645 and CCF 10-18185.}}
\author{Ersen Ekrem \qquad Sennur Ulukus
\\
\normalsize Department of Electrical and Computer Engineering\\
\normalsize University of Maryland, College Park, MD 20742 \\
\normalsize {\it ersen@umd.edu} \qquad {\it ulukus@umd.edu}}

\newcommand{\bblambda}{\bm \Lambda}
\newcommand{\brho}{\bm \rho}

\newcommand{\bbdelta}{\bm \Delta}
\newcommand{\bbsigma}{\bm \Sigma}

\newcommand{\bbi}{{\mathbf{I}}}
\newcommand{\bzero}{{\mathbf{0}}}
\newcommand{\bbv}{{\mathbf{V}}}

\newcommand{\bbh}{{\mathbf{H}}}
\newcommand{\bbw}{{\mathbf{W}}}
\newcommand{\bbm}{{\mathbf{M}}}
\newcommand{\bbq}{{\mathbf{Q}}}
\newcommand{\bbr}{{\mathbf{R}}}
\newcommand{\bbk}{{\mathbf{K}}}
\newcommand{\bz}{{\mathbf{z}}}
\newcommand{\bbz}{{\mathbf{Z}}}
\newcommand{\bbn}{{\mathbf{N}}}

\newcommand{\bba}{{\mathbf{A}}}
\newcommand{\bbd}{{\mathbf{D}}}

\newcommand{\bbt}{{\mathbf{T}}}
\newcommand{\bbb}{{\mathbf{B}}}
\newcommand{\bbc}{{\mathbf{C}}}

\newcommand{\bu}{{\mathbf{u}}}

\newcommand{\bbf}{{\mathbf{F}}}
\newcommand{\bbj}{{\mathbf{J}}}
\newcommand{\bbu}{{\mathbf{U}}}
\newcommand{\bx}{{\mathbf{x}}}
\newcommand{\bbx}{{\mathbf{X}}}
\newcommand{\by}{{\mathbf{y}}}
\newcommand{\bby}{{\mathbf{Y}}}

\def\ci{\perp\!\!\!\perp}

\newtheorem{Example}{Example}

\newtheorem{Theo}{Theorem}

\newtheorem{Lem}{Lemma}

\newtheorem{Def}{Definition}
\newenvironment{proof}[1]{\medskip\par\noindent{\bf Proof:\,}\,#1}{{\mbox{\,$\blacksquare$}\par}}

\begin{document}


\maketitle

\begin{abstract}
We study the secure lossy transmission of a vector Gaussian source
to a legitimate user in the presence of an eavesdropper, where
both the legitimate user and the eavesdropper have vector Gaussian
side information. The aim of the transmitter is to describe the
source to the legitimate user in a way that the legitimate user
can reconstruct the source within a certain distortion level while
the eavesdropper is kept ignorant of the source as much as
possible as measured by the equivocation. We obtain an outer bound
for the rate, equivocation and distortion region of this secure
lossy transmission problem. This outer bound is tight when the
transmission rate constraint is removed. In other words, we obtain
the maximum equivocation at the eavesdropper when the legitimate
user needs to reconstruct the source within a fixed distortion
level while there is no constraint on the transmission rate. This
characterization of the maximum equivocation involves two
auxiliary random variables. We show that a non-trivial selection
for both random variables may be necessary in general. The
necessity of two auxiliary random variables also implies that, in
general, Wyner-Ziv coding is suboptimal in the presence of an
eavesdropper. In addition, we show that, even when there is no
rate constraint on the legitimate link, uncoded transmission
(deterministic or stochastic) is suboptimal; the presence of an
eavesdropper necessitates the use of a coded scheme to attain the
maximum equivocation.

\end{abstract}

\newpage

\section{Introduction}

Information theoretic secrecy was initiated by Wyner
in~\cite{Wyner}, where he studied the secure lossless transmission
of a source over a degraded wiretap channel, and obtained the
necessary and sufficient conditions. Later, his result was
generalized to arbitrary, i.e., {\it not necessarily degraded},
wiretap channels in~\cite{Korner}. In recent years, information
theoretic secrecy has gathered a renewed interest, where mostly
channel coding aspects of secure transmission is considered, in
other words, secure transmission of uniformly distributed messages
is studied.

{\it Secure} source coding problem has been studied for both
lossless and lossy reconstruction cases
in~\cite{Deniz_Secure_Source_Coding_1,Deniz_Secure_Source_Coding_2,Vinod_Secure_Source_Coding,Ravi_Secure_Source_Coding,
Ravi_Secure_Source_Coding_J,Kundur_Secure_Source_Coding,Grokop_Secure_Source_Coding,Cuff_Secure_RD,Yamamoto_Secure_RD_1,Yamamoto_Secure_RD_2,Yamamoto_Secure_RD_3,Merhav_Secure_RD_1,Merhav_Secure_RD_2,Secure_Lossy_SC_Villard}.
Secure {\it lossless} source coding problem is studied
in~\cite{Deniz_Secure_Source_Coding_1,Deniz_Secure_Source_Coding_2,Vinod_Secure_Source_Coding,Ravi_Secure_Source_Coding,
Ravi_Secure_Source_Coding_J,Kundur_Secure_Source_Coding,Grokop_Secure_Source_Coding}.
The common theme of these works is that the legitimate receiver
wants to reconstruct the source in a lossless fashion by using the
information it gets from the transmitter in conjunction with its
side information, while the eavesdropper is being kept ignorant of
the source as much as possible. Secure {\it lossy} source coding
problem is studied
in~\cite{Cuff_Secure_RD,Yamamoto_Secure_RD_1,Yamamoto_Secure_RD_2,Yamamoto_Secure_RD_3,Merhav_Secure_RD_1,Merhav_Secure_RD_2,Secure_Lossy_SC_Villard}.
In these works, unlike the ones focusing on secure lossless source
coding, the legitimate receiver does not want to reconstruct the
source in a lossless fashion, but within a distortion level.

The most relevant works to our work here
are~\cite{Merhav_Secure_RD_2,Secure_Lossy_SC_Villard}.
In~\cite{Merhav_Secure_RD_2}, the author considers the secure
lossy transmission of a source over a degraded wiretap channel
while both the legitimate receiver and the eavesdropper have side
information about the source. In~\cite{Merhav_Secure_RD_2}, in
addition to the degradedness that the wiretap channel exhibits,
the source and side information also have a degradedness structure
such that given the legitimate user's side information, the source
and the eavesdropper's side information are independent. For this
setting, in~\cite{Merhav_Secure_RD_2}, a single-letter
characterization of the distortion and equivocation region is
provided. In particular, the optimality of a separation-based
approach, i.e., the optimality of a code that concatenates a
rate-distortion code and a wiretap channel code, is shown.
In~\cite{Secure_Lossy_SC_Villard}, the setting
of~\cite{Merhav_Secure_RD_2} is partially generalized such that
in~\cite{Secure_Lossy_SC_Villard}, the source and side information
do not have any degradedness structure. On the other hand, as
opposed to the {\it noisy} wiretap channel
of~\cite{Merhav_Secure_RD_2}, in~\cite{Secure_Lossy_SC_Villard},
the channel between the transmitter and receivers is assumed to be
{\it noiseless}. For this setting,
in~\cite{Secure_Lossy_SC_Villard}, a single-letter
characterization of the rate, equivocation, and distortion region
is provided.

Here, we consider the setting of~\cite{Secure_Lossy_SC_Villard}
for jointly Gaussian source and side information. In particular,
we consider the model where the transmitter has a vector Gaussian
source which is jointly Gaussian with the vector Gaussian side
information of both the legitimate receiver and the eavesdropper.
In this model, the transmitter wants to convey information to the
legitimate user in a way that the legitimate user can reconstruct
the source within a distortion level while the eavesdropper is
being kept ignorant of the source as much as possible as measured
by the equivocation. A single-letter characterization of the rate,
equivocation, and distortion region for this setting exists due
to~\cite{Secure_Lossy_SC_Villard}. Although we are unable to
evaluate this single-letter characterization for the vector
Gaussian source and side information case to obtain the
corresponding rate, equivocation, distortion region explicitly, we
obtain an outer bound for this region. We obtain this outer bound
by optimizing the rate and equivocation constraints separately. We
note that a joint optimization of the rate and equivocation
constraints for a fixed distortion level would yield the exact
achievable rate and equivocation region for this fixed distortion
level. Thus, optimizing the rate and equivocation constraints
separately yields a larger region, i.e., an outer bound. We show
that this outer bound is tight when we remove the rate constraint
at the transmitter. In other words, we obtain the maximum
achievable equivocation at the eavesdropper when the legitimate
user needs to reconstruct the vector Gaussian source within a
fixed distortion while there is no constraint on the transmission
rate.

We note some implications of this result. First, we note that
since there is no rate constraint on the transmitter, it can use
an uncoded scheme to describe the source to the legitimate user,
and, indeed, it can use any instantaneous (deterministic or
stochastic) encoding scheme for this purpose. However, we show
through an example that even when there is no rate constraint on
the transmitter, to attain the maximum equivocation at the
eavesdropper, in general, the transmitter needs to use a coded
scheme. Hence, the presence of an eavesdropper necessitates the
use of a coded scheme even in the absence of a rate constraint on
the transmitter. Second, we note that the maximum equivocation
expression has two different covariance matrices originating from
the presence of two auxiliary random variables in the
single-letter expression. We show through another example that
both of these covariance matrices, in other words, both of these
two auxiliary random variables, are needed in general to attain
the maximum equivocation at the eavesdropper. The necessity of two
covariance matrices, and hence two auxiliary random variables,
implies that, in general, Wyner-Ziv coding scheme~\cite{Wyner_Ziv}
is not sufficient to attain the maximum equivocation at the
eavesdropper.

\section{Secure Lossy Source Coding}

Here, we describe the secure lossy source coding problem and state
the existing results. Let $\{(X_i,Y_i,Z_i)\}_{i=1}^n$ denote
i.i.d. tuples drawn from a distribution $p(x,y,z)$. The
transmitter, the legitimate user and the eavesdropper observe
$X^n\in\mathcal{X}^n,Y^n\in\mathcal{Y}^n,$ and
$Z^n\in\mathcal{Z}^n$, respectively. The transmitter wants to
convey information to the legitimate user in a way that the
legitimate user can reconstruct the source $X^n$ within a certain
distortion, and meanwhile the eavesdropper is kept ignorant of the
source $X^n$ as much as possible as measured by the equivocation.
We note that if there was no eavesdropper, this setting would
reduce to the Wyner-Ziv problem~\cite{Wyner_Ziv}, for which a
single-letter characterization for the minimum transmission rate
of the transmitter for each distortion level exists.

The distortion of the reconstructed sequence at the legitimate
user is measured by the function $d^n(X^n,\hat{X}^n)$ where
$\hat{X}^n\in\mathcal{\hat{X}}^n$ denotes the legitimate user's
reconstruction of the source $X^n$. We consider the function
$d^n(X^n,\hat{X}^n)$ that has the following form
\begin{align}
d^n(X^n,\hat{X}^n)=\frac{1}{n}\sum_{i=1}^n d(X_i,\hat{X}_i)
\end{align}
where $d(a,b)$ is a non-negative finite-valued function. The
confusion of the eavesdropper is measured by the following
equivocation term
\begin{align}
\frac{1}{n}H(X^n|Z^n,M)
\end{align}
where $M\in\mathcal{M}$, which is a function of the source $X^n$,
denotes the signal sent by the transmitter.

An $(n,R)$ code for secure lossy source coding consists of an
encoding function $f_n:\mathcal{X}^n\rightarrow
\mathcal{M}=\{1,\ldots,2^{nR}\}$ at the transmitter and a decoding
function at the legitimate user $g_n:\mathcal{M}\times
\mathcal{Y}^n\rightarrow\mathcal{\hat{X}}^n$. A rate, equivocation
and distortion tuple $(R,R_e,D)$ is achievable if there exists an
$(n,R)$ code satisfying
\begin{align}
\lim_{n\rightarrow\infty} \frac{1}{n} H(X^n|Z^n,M)&\geq R_e \\
\lim_{n\rightarrow \infty} E[d(X^n,\hat{X}^n)]&\leq D
\end{align}
The set of all achievable $(R,R_e,D)$ tuples is denoted by
$\mathcal{R}^*$ which is given by the following theorem.
\begin{Theo}{\bf (\!\!\cite[Theorem 1]{Secure_Lossy_SC_Villard})}
\label{theorem_general} $(R,R_e,D)\in\mathcal{R}^*$ iff
\begin{align}
R&\geq I(V;X|Y) \label{constraint_1}\\
R_e&\leq H(X|V,Y)+I(X;Y|U)-I(X;Z|U)\label{constraint_2}\\
D&\geq E[d(X,\hat{X}(V,Y))] \label{constraint_3}
\end{align}
for some $U,V$ satisfying the following Markov chain
\begin{align}
U\rightarrow V\rightarrow X \rightarrow Y,Z
\end{align}
and a function $\hat{X}(V,Y)$.
\end{Theo}

The achievable scheme that attains the region $\mathcal{R}^*$ has
the same spirit as the Wyner-Ziv scheme~\cite{Wyner_Ziv} in the
sense that both achievable schemes use binning to exploit the side
information at the legitimate user, and consequently, to reduce
the rate requirement. The difference of the achievable scheme that
attains $\mathcal{R}^*$ comes from the additional binning
necessitated by the presence of an eavesdropper. In particular,
the transmitter generates sequences $(U^n,V^n)$ and bins both
sequences. The transmitter sends these two bin indices. Using
these bin indices, the legitimate user identifies the right
$(U^n,V^n)$ sequences, and reconstructs $X^n$ within the required
distortion. On the other hand, using the bin indices of
$(U^n,V^n)$, the eavesdropper identifies only the right $U^n$
sequence, and consequently, $U$ does not contribute to the
equivocation, see (\ref{constraint_2})\footnote{The fact that the
eavesdropper can decode $U^n$ sequence can be obtained by
observing that for a $(U,V)$ selection, if $I(U;Y)\geq I(U;Z)$,
there is no loss of optimality of setting $U=\phi$ which will
yield a larger region.}. Indeed, this achievable scheme can be
viewed as if it is using a rate-splitting technique to send the
message $M$, since $M$ has two coordinates, one for the bin index
of $U^n$, and one for the bin index of $V^n$. This perspective
reveals the similarity of the achievable scheme that attains
$\mathcal{R}^*$ and the one that attains the capacity-equivocation
region of the wiretap channel~\cite{Korner} where also
rate-splitting is used. In particular, in the latter case, the
message $W$ is divided into two parts $W_{ne},W_e$ such that
$W_{ne}$ is sent by the sequence $U^n$ and $W_e$ is sent by the
sequence $V^n$. The eavesdropper decodes $W_{ne}$ whereas the
other message $W_{e}$ contributes to the secrecy.

We note that Theorem~\ref{theorem_general} holds for continuous
$(X^n,Y^n,Z^n)$ by replacing the discrete entropy term $H(X|V,Y)$
with the differential entropy term $h(X|V,Y)$. To avoid the
negative equivocation that might arise because of the use of
differential entropy, we replace equivocation with the mutual
information leakage to the eavesdropper $I_e$ defined by
\begin{align}
\lim_{n\rightarrow \infty}\frac{1}{n} I(X^n;Z^n,M)
\end{align}
Once we are interested in the mutual information leakage to the
eavesdropper, a rate, mutual information leakage, and distortion
$(R,I_e,D)$ tuple is said to be achievable if there exists an
$(n,R)$ code such that
\begin{align}
\lim_{n\rightarrow\infty} \frac{1}{n}I(X^n;Z^n,M)&\leq I_e \\
\lim_{n\rightarrow \infty} E[d(X^n,\hat{X}^n)]&\leq D
\end{align}
The set of all achievable $(R,I_e,D)$ tuples is denoted by
$\mathcal{R}$. Using Theorem~\ref{theorem_general}, the region
$\mathcal{R}$ can be stated as follows.
\begin{Theo}{\bf (\!\!\cite{Secure_Lossy_SC_Villard})}
\label{theorem_alternative} $(R,I_e,D)\in\mathcal{R}$ iff
\begin{align}
R&\geq I(V;X|Y) \label{constraint_alt_1}\\
I_e&\geq I(V;X)-I(V;Y|U)+I(X;Z|U)\label{constraint_alt_2}\\
D&\geq E[d(X,\hat{X}(V,Y))] \label{constraint_alt_3}
\end{align}
for some $U,V$ satisfying the following Markov chain
\begin{align}
U\rightarrow V\rightarrow X \rightarrow Y,Z
\end{align}
and a function $\hat{X}(V,Y)$.
\end{Theo}

\section{Vector Gaussian Sources}

Now we study the secure lossy source coding problem for jointly
Gaussian $\{(\bbx_i,\bby_i,\bbz_i)\}_{i=1}^n$ where the tuples
$\{(\bbx_i,\bby_i,\bbz_i)\}_{i=1}^n$ are independent across time,
i.e., across the index $i$, and each tuple is drawn from the same
jointly Gaussian distribution $p(\bbx,\bby,\bbz)$. In other words,
we consider the case where $\bbx_i$ is a zero-mean Gaussian random
vector with covariance matrix $\bbk_X\succ \bzero $, and the side
information at the legitimate user $\bby_i$ and the eavesdropper
$\bbz_i$ are jointly Gaussian with the source $\bbx_i$. In
particular, we assume that $\bby_i,\bbz_i$ have the following form
\begin{align}
\bby_i&=\bbx_i+\bbn_{Y,i} \label{aligned_1}\\
\bbz_i&=\bbx_i+\bbn_{Z,i} \label{aligned_2}
\end{align}
where $\bbn_{Y,i}$ and $\bbn_{Z,i}$ are independent zero-mean
Gaussian random vectors with covariance matrices $\bbsigma_Y\succ
\bzero $ and $\bbsigma_Z\succ \bzero$, respectively, and
$(\bbn_{Y,i},\bbn_{Z,i})$ and $\bbx_{i}$ are independent. We note
that the side information given by
(\ref{aligned_1})-(\ref{aligned_2}) are not in the most general
form. In the most general case, we have
\begin{align}
\bby_{i}&=\bbh_{Y}\bbx_{i}+\bbn_{Y,i}\label{most_1}\\
\bbz_{i}&=\bbh_{Z}\bbx_{i}+\bbn_{Z,i} \label{most_2}
\end{align}
for some $\bbh_Y,\bbh_Z$ matrices. However, until Section
\ref{sec:general_case}, we consider the form of side information
given by (\ref{aligned_1})-(\ref{aligned_2}), and obtain our
results for this model. In Section~\ref{sec:general_case}, we
generalize our results to the most general case given by
(\ref{most_1})-(\ref{most_2}). We note that since the rate,
information leakage and distortion region is invariant with
respect to the correlation between $\bbn_{Y,i}$ and $\bbn_{Z,i}$,
the correlation between $\bbn_{Y,i}$ and $\bbn_{Z,i}$ is
immaterial.

The distortion of the reconstructed sequence
$\{\hat{\bbx}_i\}_{i=1}^n$ is measured by the mean square error
matrix:
\begin{align}
E\left[\big(\bbx_i-\hat{\bbx}_i\big)\big(\bbx_i-\hat{\bbx}_i\big)^\top\right]
\end{align}
Hence, the distortion constraint is represented by a positive
semi-definite matrix $\bbd$, which is achievable if there is an
$(n,R)$ code such that
\begin{align}
\frac{1}{n}\sum_{i=1}^nE\left[\big(\bbx_i-\hat{\bbx}_i\big)\big(\bbx_i-\hat{\bbx}_i\big)^\top\right]\preceq
\bbd
\end{align}
Throughout the paper, we assume that $\bzero\preceq \bbd \preceq
\bbk_{X|Y}$. Since the mean square error is minimized by the
minimum mean square error (MMSE) estimator which is given by the
conditional mean, we assume that the legitimate user applies this
optimal estimator, i.e., the legitimate user selects its
reconstruction function $\{\hat{\bbx}_i\}_{i=1}^n$ as
\begin{align}
\hat{\bbx}_i=E\left[\bbx_i|\bby^n,f_n(\bbx^n)\right]
\label{estimator}
\end{align}
Once the estimator of the legitimate user is set as
(\ref{estimator}), using Theorem~\ref{theorem_alternative}, a
single-letter description of the region $\mathcal{R}$ for a vector
Gaussian source can be given as follows.
\begin{Theo}
\label{theorem_Gauss_single} $(R,I_e,\bbd)\in\mathcal{R}$ iff
\begin{align}
R&\geq I(V;\bbx|\bby) \label{rate_constraint} \\
I_e&\geq I(V;\bbx)-I(V;\bby|U)+I(\bbx;\bbz|U)\label{leakage_constraint}\\
\bbd&\succeq  \bbk_{X|VY}
\end{align}
for some $U,V$ satisfying the following Markov chain
\begin{align}
U\rightarrow V\rightarrow \bbx \rightarrow \bby,\bbz
\end{align}
\end{Theo}
We also define the region $\mathcal{R}(\bbd)$ as the union of the
$(R,I_e)$ pairs that are achievable when the distortion constraint
matrix is set to $\bbd$. Our main result is an outer bound for the
region $\mathcal{R}(\bbd)$, hence for the region $\mathcal{R}$.
\begin{Theo}
\label{theorem_outer} When $\bbd \preceq \bbk_{X|Y}$, we have
\begin{align}
\mathcal{R}(\bbd) \subseteq \mathcal{R}^{o}(\bbd)
\end{align}
where $\mathcal{R}^{o}(\bbd)$ is given by the union of $(R,I_e)$
that satisfy
\begin{align}
R&\geq \frac{1}{2} \log
\frac{|\bbk_{X|Y}|}{|\bbd|}=\frac{1}{2}\log\frac{|\bbk_X|}{|\bbf(\bbd)|}-\frac{1}{2}\log\frac{|\bbk_X+\bbsigma_Y|}{|\bbf(\bbd)+\bbsigma_Y|}
\label{rate_bound}\\
I_e & \geq~\min_{\substack{\bzero\preceq \bbk_{X|V}\preceq
\bbk_{X|U}\preceq \bbk_X\\ \bbk_{X|V}\preceq \bbf(\bbd)}}
\frac{1}{2}\log\frac{|\bbk_X|}{|\bbk_{X|V}|}-\frac{1}{2}\log\frac{|\bbk_{X|U}+\bbsigma_Y|}{|\bbk_{X|V}+\bbsigma_Y|}
+\frac{1}{2} \log\frac{|\bbk_{X|U}+\bbsigma_Z|}{|\bbsigma_Z|}
\label{leakage_bound}
\end{align}
and
$\bbf(\bbd)=\bbsigma_Y(\bbsigma_Y-\bbd)^{-1}\bbsigma_Y-\bbsigma_Y$.
\end{Theo}
We will prove Theorem~\ref{theorem_outer} in
Section~\ref{sec:proof_of_theorem_outer}. In the remainder of this
section, we provide interpretations and discuss some implications
of Theorem~\ref{theorem_outer}.

The outer bound in Theorem~\ref{theorem_outer} is obtained by
minimizing the constraints on $R$ and $I_e$ individually, i.e.,
the rate lower bound in (\ref{rate_bound}) is obtained by
minimizing the rate constraint in (\ref{rate_constraint}) and the
mutual information leakage lower bound in (\ref{leakage_bound}) is
obtained by minimizing the mutual information leakage constraint
in (\ref{leakage_constraint}) separately. However, to characterize
the rate and mutual information leakage region
$\mathcal{R}(\bbd)$, one needs to minimize the rate constraint in
(\ref{rate_constraint}) and the mutual leakage information
constraint in (\ref{leakage_constraint}) jointly, not separately.
In particular, since the region $\mathcal{R}(\bbd)$ is convex in
the pairs $(R,I_e)$ as per a time-sharing argument, joint
optimization of the rate constraint in (\ref{rate_constraint}) and
the mutual information leakage constraint in
(\ref{leakage_constraint}) can be carried out by considering the
tangent lines to the region $\mathcal{R}(\bbd)$, i.e., by solving
the following optimization problem
\begin{align}
L(\mu_1,\mu_2)&=\min_{(R,I_e)\in\mathcal{R}(\bbd)}~\mu_1R+\mu_2I_e\\
&=\min_{\substack{U\rightarrow V\rightarrow \bbx \rightarrow
\bby,\bbz\\\bbk_{X|VY}\preceq \bbd}}
\mu_1\left[I(V;\bbx)-I(V;\bby)\right]+\mu_2\left[I(V;\bbx)-I(V;\bby|U)+I(\bbx;\bbz|U)\right]
\label{tangent_lines}
\end{align}
for all values of $\mu_1,\mu_2$, where
$\mu_j\in[0,\infty),~j=1,2$. As of now, we have been unable to
solve the optimization problem $L(\mu_1,\mu_2)$ for all values of
$(\mu_1,\mu_2)$. However, as stated in
Theorem~\ref{theorem_outer}, we solve the optimization problems
$L(0,\mu_2)$ and $L(\mu_1,0)$ by showing the optimality of jointly
Gaussian $(U,V,\bbx)$ to evaluate the corresponding cost
functions. In other words, our outer bound in
Theorem~\ref{theorem_outer} can be written as follows.
\begin{align}
R&\geq L(1,0)\\
I_e&\geq L(0,1)
\end{align}
We note that the constraint in (\ref{rate_bound}), and hence
$L(1,0)$, gives us the Wyner-Ziv rate distortion
function~\cite{Wyner_Ziv} for the vector Gaussian sources.
Moreover, we note that $L(0,1)$ gives us the minimum mutual
information leakage to the eavesdropper when the legitimate user
wants to reconstruct the source within a fixed distortion
constraint $\bbd$ while there is no concern on the transmission
rate $R$. Denoting the minimum mutual information leakage to the
eavesdropper when the legitimate user needs to reconstruct the
source within a fixed distortion constraint $\bbd$ by
$I_e^{\min}(\bbd)$, the corresponding result can be stated as
follows.
\begin{Theo}
\label{theorem_Ie_min} When $\bbd \preceq \bbk_{X|Y}$, we have
\begin{align}
I_e^{\min}(\bbd)=~\min_{\substack{\bzero\preceq \bbk_{X|V}\preceq
\bbk_{X|U}\preceq \bbk_X\\ \bbk_{X|V}\preceq \bbf(\bbd)}}
\frac{1}{2}\log\frac{|\bbk_X|}{|\bbk_{X|V}|}-\frac{1}{2}\log\frac{|\bbk_{X|U}+\bbsigma_Y|}{|\bbk_{X|V}+\bbsigma_Y|}
+\frac{1}{2} \log\frac{|\bbk_{X|U}+\bbsigma_Z|}{|\bbsigma_Z|}
\label{Ie_min_optimization}
\end{align}
where
$\bbf(\bbd)=\bbsigma_Y(\bbsigma_Y-\bbd)^{-1}\bbsigma_Y-\bbsigma_Y$.
\end{Theo}

Theorem~\ref{theorem_Ie_min} implies that if the transmitter's aim
is to minimize the mutual information leakage to the eavesdropper
without concerning itself with the rate it costs as long as the
legitimate receiver is able to reconstruct the source within a
distortion constraint $\bbd$, the use of jointly Gaussian
$(U,V,\bbx)$ is optimal. Since in Theorem~\ref{theorem_Ie_min},
there is no rate constraint, one natural question to ask is
whether $I_e^{\rm min}(\bbd)$ can be achieved by an uncoded
transmission scheme. Now, we address this question in a broader
context by letting the encoder use any {\it instantaneous}
encoding function in the form of $g_{i}(\bbx_i)$ where
$g_i(\cdot)$ can be a deterministic or a stochastic mapping. When
$g_i(\cdot)$ is chosen to be stochastic, we assume it to be
independent across time. We note that the uncoded transmission can
be obtained from instantaneous encoding by selecting $g_i(\cdot)$
to be a linear function. Similarly, uncoded transmission with
artificial noise can be obtained from instantaneous encoding by
selecting $g_i(x)=\alpha x+ N$, where $N$ denotes the noise.
Hence, if the encoder uses an instantaneous encoding scheme, the
transmitted signal is given by
$M=\left[~g_1(\bbx_1),\ldots,g_n(\bbx_n)~\right]$. Let $I_e^{\rm
ins}(\bbd)$ be the minimum information leakage to the eavesdropper
when the legitimate user is able to reconstruct the source with a
distortion constraint $\bbd$ while the encoder uses an
instantaneous encoding. The following example demonstrates that,
in general, $I_e^{\rm min}(\bbd)$ cannot be achieved by
instantaneous encoding.

\begin{Example} Consider the scalar case, where the side
information at the legitimate user and the eavesdropper are given
as follows
\begin{align}
Y_i&=X_i+N_{y,i} \label{scalar_y}\\
Z_i&=X_i+N_{z,i} \label{scalar_z}
\end{align}
where $X_i,N_{y,i}$ and $N_{z,i}$ are zero-mean Gaussian random
variables with variances $\sigma_x^2,\sigma_y^2$ and $\sigma_z^2$,
respectively. $\{X_i\}_{i=1}^n,\{N_{y,i}\}_{i=1}^n$ and
$\{N_{z,i}\}_{i=1}^n$ are independent. We assume that
$\sigma_y^2<\sigma_z^2$, which implies that we can assume
$X\rightarrow Y\rightarrow Z$ since the scalar model in
(\ref{scalar_y})-(\ref{scalar_z}) is statistically degraded, or in
other words, the correlation between $N_{y,i}$ and $N_{z,i}$ does
not affect the achievable $(R,I_e,D)$ region. Using
Theorem~\ref{theorem_Gauss_single}, $I_e^{\rm min}(D)$ for the
scalar Gaussian channel under consideration can be found as
follows
\begin{align}
I_e^{\rm min}(D)&=\min_{\substack{U\rightarrow V\rightarrow
X\rightarrow Y \rightarrow Z\\ \sigma_{x|vy}^2 \leq D}}
I(V;X)-I(V;Y|U)+I(X;Z|U) \\
&=\min_{\substack{ V\rightarrow X\rightarrow Y \rightarrow Z\\
\sigma_{x|vy}^2 \leq D}} I(V;X)-I(V;Y)+I(X;Z)
\label{mc_degraded_implies_again}
\end{align}
where in (\ref{mc_degraded_implies_again}), we used the Markov
chain $U\rightarrow V\rightarrow X\rightarrow Y\rightarrow Z$.

As shown in Appendix~\ref{proof_of_I_e_ins_0}, the information
leakage to the eavesdropper when the encoder uses an instantaneous
mapping is given by
\begin{align}
I_e^{\rm ins}(D)&=\min_{\substack{V\rightarrow X\rightarrow
Y\rightarrow Z\\\sigma_{x|vy}^2\leq D}} I(X;V,Z)\label{I_e_inst_min_0}\\
&=\min_{\substack{V\rightarrow X\rightarrow Y\rightarrow
Z\\\sigma_{x|vy}^2\leq D}} I(V;X)-I(V;Z)+I(X;Z)
\label{I_e_inst_min}
\end{align}
where (\ref{I_e_inst_min}) is obtained by using the Markov chain
$V\rightarrow X\rightarrow Z$.

Using (\ref{mc_degraded_implies_again}) and (\ref{I_e_inst_min}),
we have
\begin{align}
I_e^{\rm ins}(D)-I_e^{\rm min}(D)&=\min_{\substack{V\rightarrow
X\rightarrow Y\rightarrow Z\\\sigma_{x|vy}^2\leq D}}
I(V;X)-I(V;Z)+I(X;Z)
\nonumber\\
&\quad - \min_{\substack{ V\rightarrow X\rightarrow Y \rightarrow Z\\
\sigma_{x|vy}^2 \leq D}} I(V;X)-I(V;Y)+I(X;Z)\\
&\geq \min_{\substack{V\rightarrow X\rightarrow Y\rightarrow
Z\\\sigma_{x|vy}^2\leq D}} I(V;Y)-I(V;Z)\\
&= \min_{\substack{V\rightarrow X\rightarrow Y\rightarrow
Z\\\sigma_{x|vy}^2\leq D}} I(V;Y|Z) \label{some_mc_should_0}
\end{align}
where (\ref{some_mc_should_0}) comes from the Markov chain
$V\rightarrow Y\rightarrow Z$. Next, we note the following lemma.
\begin{Lem}
\label{lemma_strictly_nonzero} For jointly Gaussian $(X,Y,Z)$
satisfying the Markov chain $X\rightarrow Y\rightarrow Z$ and
$\Pr[Y=Z]\neq 1$, if $D<\sigma^2_{x|y}$, we have
\begin{align}
\min_{\substack{V\rightarrow X\rightarrow Y\rightarrow
Z\\\sigma_{x|vy}^2\leq D}} I(V;Y|Z)>0 \label{some_mc_should_00}
\end{align}
\end{Lem}
The proof of Lemma~\ref{lemma_strictly_nonzero} can be found in
Appendix~\ref{proof_of_lemma_strictly_nonzero}. The proof of
Lemma~\ref{lemma_strictly_nonzero} starts with the observation
that (\ref{some_mc_should_00}) is zero iff we have the Markov
chain $V\rightarrow Z\rightarrow Y$. On the other hand, since we
already have the Markov chain $V\rightarrow X\rightarrow
Y\rightarrow Z $, and $Y$ and $Z$ are not identical, we show in
Appendix~\ref{proof_of_lemma_strictly_nonzero} that the Markov
chain $V\rightarrow Z\rightarrow Y$ is possible iff $V$ and $X$
are independent. However, if $D<\sigma_{x|y}^2$, any $V$ that is
independent of $X$ is not feasible. Hence,
Lemma~\ref{lemma_strictly_nonzero} follows.
Lemma~\ref{lemma_strictly_nonzero} implies that in general, we
have $I_e^{\rm ins}(D)\neq I_e^{\rm min} (D)$, i.e., $I_e^{\rm
min}(\bbd)$ cannot be achieved by instantaneous encoding.
\end{Example}
This example shows that an uncoded transmission is not optimal
even when there is no rate constraint. This is due to the presence
of an eavesdropper; the presence of an eavesdropper necessitates
the use of a coded scheme.

Another question that Theorem~\ref{theorem_Ie_min} brings about is
whether the minimum in (\ref{Ie_min_optimization}) is achieved by
a non-trivial $\bbk_{X|U}$. By a trivial selection for
$\bbk_{X|U}$ we mean either $\bbk_{X|U}=\bbk_X$ or
$\bbk_{X|U}=\bbk_{X|V}$. The former corresponds to the selection
$U=\phi$ and the latter corresponds to the selection $U=V$. We
note that although (\ref{Ie_min_optimization}) is monotonically
decreasing in $\bbk_{X|V}$ in the positive semi-definite sense,
(\ref{Ie_min_optimization}) is neither monotonically increasing
nor monotonically decreasing in $\bbk_{X|U}$ in the positive
semi-definite sense. Hence, due to this lack of monotonicity of
$(\ref{Ie_min_optimization})$ in $\bbk_{X|U}$, in general, we
expect that both $U\neq \phi$ and $U\neq V$ may be necessary to
attain the minimum in (\ref{Ie_min_optimization}). The following
example demonstrates that in general $U\neq \phi$ and $U\neq V$
may be necessary.
\begin{Example}
\label{example_double_aux} Consider the Gaussian source
$\bbx=[~X_1~~X_2~]^\top$ where $X_1$ and $X_2$ are independent.
The side information at the legitimate receiver and the
eavesdropper are given by
\begin{align}
Y_\ell &=X_{\ell}+N_{Y,\ell},\quad \ell=1,2 \\
Z_\ell &=X_{\ell}+N_{Z,\ell},\quad \ell=1,2
\end{align}
where $N_{Y,\ell}$ and $N_{Z,\ell}$ are zero-mean Gaussian random
variables with variances $\sigma_{Y,\ell}^2$ and
$\sigma_{Z,\ell}^2$, respectively. Moreover, $N_{Y,1}$ and
$N_{Y,2}$ are independent, and also so are $N_{Z,1}$ and
$N_{Z,2}$. We assume that noise variances satisfy
\begin{align}
\sigma_{Y,1}^2 &< \sigma_{Z,1}^2 \\
\sigma_{Z,2}^2 &< \sigma_{Y,2}^2
\end{align}
which, in view of the fact that correlation between the noise at
the legitimate receiver and the noise at the eavesdropper does not
affect the rate, distortion and information leakage region, lets
us assume the following Markov chains
\begin{align}
X_1\rightarrow Y_1 \rightarrow Z_1 \\
X_2\rightarrow Z_2 \rightarrow Y_2
\end{align}
Moreover, we assume that the distortion constraint $\bbd$ is a
diagonal matrix with diagonal entries $D_1$ and $D_2$. In this
case, the minimum information leakage is given by
\begin{align}
I_e^{\rm min}(D_1,D_2)&=\min_{\substack{V_1\rightarrow
X_1\rightarrow Y_1\rightarrow Z_1\\\sigma_{X_1|V_1Y_1}^2\leq
D_1}}~I(V_1;X_1)-I(V_1;Y_1)+I(X_1;Z_1)\nonumber\\
&\quad +\min_{\substack{V_2\rightarrow X_2\rightarrow
Z_2\rightarrow Y_2\\\sigma_{X_2|V_2Y_2}^2\leq
D_2}}~I(V_2;X_2)+I(X_2;Z_2|V_2) \label{minimum_leakage_parallel}
\end{align}
whose proof can be found in
Appendix~\ref{proof_of_minimum_leakage_parallel}. The minimum
information leakage in (\ref{minimum_leakage_parallel})
corresponds the selections $U=(\phi,V_2)$ and $V=(V_1,V_2)$, where
$(U_1,V_1)$ and $(U_2,V_2)$ are independent. This selection of
$(U,V)$ corresponds to neither $U=\phi$ nor $U=V$.

Next, we obtain the minimum information leakage that arises when
we set either $U=\phi$ or $U=V$, and show that the minimum
information leakage arising from these selections are strictly
larger than the minimum information leakage in
(\ref{minimum_leakage_parallel}), which will imply the
suboptimality of $U=\phi$ and $U=V$. When we set $U=\phi$, the
minimum information leakage is given by
\begin{align}
I_e^{\rm min-\phi}(D_1,D_2)&= \min_{\substack{V_1\rightarrow
X_1\rightarrow Y_1\rightarrow Z_1\\\sigma_{X_1|V_1Y_1}^2\leq
D_1}}~I(V_1;X_1)-I(V_1;Y_1)+I(X_1;Z_1)\nonumber\\
&\quad +\min_{\substack{V_2\rightarrow X_2\rightarrow
Z_2\rightarrow Y_2\\\sigma_{X_2|V_2Y_2}^2\leq
D_2}}~I(V_2;X_2)-I(V_2;Y_2)+I(X_2;Z_2) \label{minimum_leakage_phi}
\end{align}
whose proof is given in
Appendix~\ref{proof_of_minimum_leakage_single}. When we set $U=V$,
the minimum information leakage is given by
\begin{align}
I_e^{\rm min-S}(D_1,D_2)&=\min_{\substack{V_1\rightarrow
X_1\rightarrow Y_1\rightarrow Z_1\\\sigma_{X_1|V_1Y_1}^2\leq
D_1}}~I(V_1;X_1)+I(X_1;Z_1|V_1)\nonumber\\
&\quad +\min_{\substack{V_2\rightarrow X_2\rightarrow
Z_2\rightarrow Y_2\\\sigma_{X_2|V_2Y_2}^2\leq
D_2}}~I(V_2;X_2)+I(X_2;Z_2|V_2) \label{minimum_leakage_single}
\end{align}
whose proof can be found in
Appendix~\ref{proof_of_minimum_leakage_single}.

Now, we compare the minimum information leakage in
(\ref{minimum_leakage_parallel}) with (\ref{minimum_leakage_phi})
and (\ref{minimum_leakage_single}) to show that the selections
$U=\phi$ and $U=V$ are sub-optimal in general. Using
(\ref{minimum_leakage_parallel}) and (\ref{minimum_leakage_phi}),
we get
\begin{align}
I_e^{\rm min-\phi}(D_1,D_2)-I_e^{\rm
min}(D_1,D_2)&=\min_{\substack{V_2\rightarrow X_2\rightarrow
Z_2\rightarrow Y_2\\\sigma_{X_2|V_2Y_2}^2\leq
D_2}}~I(V_2;X_2)-I(V_2;Y_2)+I(X_2;Z_2)\nonumber\\
&\quad -\min_{\substack{V_2\rightarrow X_2\rightarrow
Z_2\rightarrow Y_2\\\sigma_{X_2|V_2Y_2}^2\leq
D_2}}~I(V_2;X_2)+I(X_2;Z_2|V_2)\\
&\geq \min_{\substack{V_2\rightarrow X_2\rightarrow Z_2\rightarrow
Y_2\\\sigma_{X_2|V_2Y_2}^2\leq
D_2}}~I(X_2;Z_2)-I(X_2;Z_2|V_2)-I(V_2;Y_2)\\
&=\min_{\substack{V_2\rightarrow X_2\rightarrow Z_2\rightarrow
Y_2\\\sigma_{X_2|V_2Y_2}^2\leq D_2}}~I(V_2;Z_2)-I(V_2;Y_2)
\label{some_mc_should} \\
&=\min_{\substack{V_2\rightarrow X_2\rightarrow Z_2\rightarrow
Y_2\\\sigma_{X_2|V_2Y_2}^2\leq D_2}}~I(V_2;Z_2|Y_2)
\label{some_mc_should_1}\\
&>0 \label{lemma_strictly_nonzero_implies}
\end{align}
where (\ref{some_mc_should})-(\ref{some_mc_should_1}) follow from
the Markov chain
\begin{align}
V_2 \rightarrow X_2\rightarrow Z_2\rightarrow Y_2
\end{align}
and (\ref{lemma_strictly_nonzero_implies}) comes from
Lemma~\ref{lemma_strictly_nonzero}. Thus, in general, we have $
I_e^{\rm min-\phi}(D_1,D_2)\neq I_e^{\rm min}(D_1,D_2)$, or in
other words, in general, $U=\phi$ is sub-optimal.

Next, we consider the selection $U=V$. Using
(\ref{minimum_leakage_parallel}) and
(\ref{minimum_leakage_single}), we have
\begin{align}
I_e^{\rm min-S}(D_1,D_2)-I_e^{\rm
min}(D_1,D_2)&=\min_{\substack{V_1\rightarrow X_1\rightarrow
Y_1\rightarrow Z_1\\\sigma_{X_1|V_1Y_1}^2\leq
D_1}}~I(V_1;X_1)+I(X_1;Z_1|V_1)\nonumber\\
&\quad -\min_{\substack{V_1\rightarrow X_1\rightarrow
Y_1\rightarrow Z_1\\\sigma_{X_1|V_1Y_1}^2\leq
D_1}}~I(V_1;X_1)-I(V_1;Y_1)+I(X_1;Z_1)\\
&\geq \min_{\substack{V_1\rightarrow X_1\rightarrow Y_1\rightarrow
Z_1\\\sigma_{X_1|V_1Y_1}^2\leq
D_1}}~I(X_1;Z_1|V_1)+I(V_1;Y_1)-I(X_1;Z_1)\\
&=\min_{\substack{V_1\rightarrow X_1\rightarrow Y_1\rightarrow
Z_1\\\sigma_{X_1|V_1Y_1}^2\leq D_1}}~I(V_1;Y_1)-I(V_1;Z_1)
\label{some_mc_should_2}\\
&=\min_{\substack{V_1\rightarrow X_1\rightarrow Y_1\rightarrow
Z_1\\\sigma_{X_1|V_1Y_1}^2\leq D_1}}~I(V_1;Y_1|Z_1)
\label{some_mc_should_3}\\
&>0\label{lemma_strictly_nonzero_implies_1}
\end{align}
where (\ref{some_mc_should_2})-(\ref{some_mc_should_3}) follow
from the Markov chain
\begin{align}
V_1\rightarrow X_1 \rightarrow Y_1 \rightarrow Z_1
\end{align}
and (\ref{lemma_strictly_nonzero_implies_1}) comes from
Lemma~\ref{lemma_strictly_nonzero}. Thus, in general, we have $
I_e^{\rm min-S}(D_1,D_2)\neq I_e^{\rm min}(D_1,D_2)$, or in other
words, in general, $U=V$ is sub-optimal.
\end{Example}

Example~\ref{example_double_aux} shows that, in general, we might
need two covariance matrices, and hence two different auxiliary
random variables, to attain the minimum information leakage.
Indeed, if we have either $U=V$ or $U=\phi$, the corresponding
achievable scheme is identical to the Wyner-Ziv
scheme~\cite{Wyner_Ziv}. Hence, the necessity of two different
auxiliary random variables implies that, in general, Wyner-Ziv
scheme~\cite{Wyner_Ziv} is suboptimal.

\section{Proof of Theorem~\ref{theorem_outer}}

\label{sec:proof_of_theorem_outer}

We now provide the proof of Theorem~\ref{theorem_outer}. As
mentioned in the previous section, this outer bound is obtained by
minimizing the rate constraint in (\ref{rate_constraint}) and the
mutual information leakage constraint in
(\ref{leakage_constraint}) separately. We first consider the rate
constraint in (\ref{rate_constraint}) as follows
\begin{align}
R&\geq L(1,0)\label{dummy_is_dummy}\\
&=\min_{\substack{V\rightarrow\bbx\rightarrow \bby,\bbz \label{future_use_0}\\
\bbk_{X|VY}\preceq \bbd}}~I(V;\bbx|\bby)\\
&= \min_{\substack{V\rightarrow\bbx\rightarrow \bby,\bbz\\
\bbk_{X|VY}\preceq \bbd}}~h(\bbx|\bby)-h(\bbx|V,\bby)\\
&=\min_{\substack{V\rightarrow\bbx\rightarrow \bby,\bbz\\
\bbk_{X|VY}\preceq \bbd}}~\frac{1}{2}\log |(2\pi e)\bbk_{X|Y}|-h(\bbx|V,\bby)\\
&=\min_{\bbk_{X|VY}\preceq \bbd}~\frac{1}{2}\log\frac{|\bbk_{X|Y}|}{|\bbk_{X|VY}|}\label{Gaussian_maximizer}\\
&=\frac{1}{2} \log\frac{|\bbk_{X|Y}|}{|\bbd|}\label{monotonocity}
\end{align}
where (\ref{Gaussian_maximizer}) comes from the fact that
$h(\bbx|V,\bby)$ is maximized by jointly Gaussian $(V,\bbx,\bby)$,
and (\ref{monotonocity}) comes from the monotonicity of $|\cdot|$
in positive semi-definite matrices. Now we introduce the following
lemma.
\begin{Lem}
\label{lemma_equivalence}
\begin{align}
\frac{1}{2}\log\frac{|\bbk_{X|Y}|}{|\bbd|}=\frac{1}{2}\log\frac{|\bbk_X|}{|\bbf(\bbd)|}-\frac{1}{2}\log\frac{|\bbk_X+\bbsigma_Y|}{|\bbf(\bbd)+\bbsigma_Y|}
\label{lemma_equivalence_statement}
\end{align}
\end{Lem}
The proof of Lemma~\ref{lemma_equivalence} is given in
Appendix~\ref{proof_of_lemma_equivalence}.
Lemma~\ref{lemma_equivalence} and (\ref{monotonocity}) imply
(\ref{rate_bound}).

Next, we consider the mutual information leakage constraint in
(\ref{leakage_constraint}) as follows
\begin{align}
I_e\geq L(0,1) =\min_{\substack{U\rightarrow V\rightarrow
\bbx\rightarrow \bby,\bbz\\\bbk_{X|VY}\preceq
\bbd}}I(V;\bbx)-I(V;\bby|U)+I(\bbx;\bbz|U) \label{Ie_min_program}
\end{align}
We note that the cost function of $L(0,1)$ can be rewritten as
follows
\begin{align}
C(L)&=I(V;\bbx)-I(V;\bby)+I(U;\bby)+I(\bbx;\bbz|U) \label{MC_implies_0}\\
&=I(V;\bbx|\bby)+\left[I(U;\bby)+I(\bbx;\bbz|U)\right]\label{MC_implies}
\end{align}
where (\ref{MC_implies_0}) comes from the Markov chain
$U\rightarrow V\rightarrow \bby$ and (\ref{MC_implies}) comes from
the Markov chain $V\rightarrow \bbx\rightarrow \bby$. We note that
the first term in (\ref{MC_implies}) is minimized by a jointly
Gaussian $(V,\bbx)$ as we already showed in obtaining the lower
bound for the rate given by (\ref{rate_bound}) above in
(\ref{dummy_is_dummy})-(\ref{monotonocity}). On the other hand,
the remaining term of (\ref{MC_implies}) in the bracket is
maximized by a jointly Gaussian $(U,\bbx)$ as shown
in~\cite{Liu_Extremal_Inequality}. Thus, a tension between these
two terms arises if $(U,V,\bbx)$ is selected to be jointly
Gaussian. In spite of this tension, we will still show that a
jointly Gaussian $(U,V,\bbx)$ is the minimizer of $L(0,1)$.
Instead of directly showing this, we first characterize the
minimum mutual information leakage when $(U,V,\bbx)$ is restricted
to be jointly Gaussian, and show that this cannot be attained by
any other distribution for $(U,V,\bbx)$. We note that any jointly
Gaussian $(U,V,\bbx)$ can be written as
\begin{align}
V&=\bba_V\bbx+\bbn_V\\
U&=\bba_U\bbx+\bbn_U
\end{align}
where $\bbn_V,\bbn_U$ are zero-mean Gaussian random vectors with
covariance matrices $\bbsigma_V,\bbsigma_U$, respectively.
Moreover, $\bbn_V,\bbn_U$  are independent of $\bbx,\bby,\bbz,$
but can be dependent on each other. Before characterizing the
minimum mutual information leakage when $(U,V,\bbx)$ is restricted
to be jointly Gaussian, we introduce the following lemma.
\begin{Lem}
\label{lemma_equivalence_1} When $\bbd\preceq \bbk_{X|Y}$ and $V$
is Gaussian, we have the following facts.
\begin{itemize}
\item $\bbsigma_Y-\bbd\succ \bzero$, i.e., $\bbsigma_Y-\bbd$ is
positive definite, and hence, non-singular.

\item We have the following equivalence:
\begin{align}
\bbk_{X|VY}\preceq \bbd ~~ \Longleftrightarrow~~ \bbk_{X|V}\preceq
\bbf(\bbd)
\end{align}
\end{itemize}
\end{Lem}
The proof of Lemma~\ref{lemma_equivalence_1} is given in
Appendix~\ref{proof_of_lemma_equivalence_1}. Using
Lemma~\ref{lemma_equivalence_1}, the minimum mutual information
leakage to the eavesdropper when $(U,V,\bbx)$ is restricted to be
jointly Gaussian can be written as follows:
\begin{align}
L^G=\min_{\substack{U\rightarrow V\rightarrow \bbx\rightarrow
\bby,\bbz\\(U,V,\bbx)\textrm{ is jointly
Gaussian}\\\bbk_{X|V}\preceq
\bbf(\bbd)}}I(V;\bbx)-I(V;\bby|U)+I(\bbx;\bbz|U)
\label{minimum_Gaussian_leakage}
\end{align}

\noindent We note that the minimization in
(\ref{minimum_Gaussian_leakage}) can be written as a minimization
of the cost function in (\ref{minimum_Gaussian_leakage}) over all
possible $\bba_U,\bba_V,\bbsigma_U,\bbsigma_V$ matrices by
expressing $\bbk_{X|U}$ and $\bbk_{X|V}$ in terms of
$\bba_U,\bba_V,\bbsigma_U,\bbsigma_V$. Instead of considering this
tedious optimization problem, we consider the following one:
\begin{align}
\bar{L}^G=\min_{\substack{\bzero \preceq \bbk_{X|V}\preceq
\bbk_{X|U}\preceq \bbk_X\\\bbk_{X|V}\preceq \bbf(\bbd)}}
\frac{1}{2}\log\frac{|\bbk_X|}{|\bbk_{X|V}|}
-\frac{1}{2}\log\frac{|\bbk_{X|U}+\bbsigma_Y|}{|\bbk_{X|V}+\bbsigma_Y|}+\frac{1}{2}\log\frac{|\bbk_{X|U}+\bbsigma_Z|}{|\bbsigma_Z|}
\label{minimum_Gaussian_leakage_lower}
\end{align}
We note that due to the Markov chain $U\rightarrow V\rightarrow
\bbx$, we always have $\bbk_{X|V}\preceq \bbk_{X|U}$. A proof of
this fact is given in
Appendix~\ref{proof_of_conditioning_reduces_mmse}. Besides this
inequality, $\bbk_{X|V}$ and $\bbk_{X|U}$ might have further
interdependencies which are not considered in the optimization
problem in (\ref{minimum_Gaussian_leakage_lower}). Since
neglecting these further interdependencies among $\bbk_{X|U}$ and
$\bbk_{X|V}$ enlarges the feasible set of the optimization problem
in (\ref{minimum_Gaussian_leakage}), we have, in general,
\begin{align}
L^G \geq \bar{L}^G
\end{align}
On the other hand, it can be shown that the value of $\bar{L}^G$
can be obtained by some jointly Gaussian $(U,V,\bbx)$ satisfying
the Markov chain $U\rightarrow V\rightarrow \bbx$, as stated in
the following lemma.
\begin{Lem}
\label{lemma_equivalence_2}
\begin{align}
L^G=\bar{L}^G
\end{align}
\end{Lem}
The proof of Lemma~\ref{lemma_equivalence_2} is given in
Appendix~\ref{proof_of_lemma_equivalence_2}.

Now we study the optimization problem $\bar{L}^G$ in
(\ref{minimum_Gaussian_leakage_lower}) in more detail. Let
$\bbk_{X|V}^*$ and $\bbk_{X|U}^*$ be the minimizers for the
optimization problem $\bar{L}^G$. They need to satisfy the
following KKT conditions.
\begin{Lem}
\label{lemma_KKT} If $\bbk_{X|V}^*$ and $\bbk_{X|U}^*$ are the
minimizers for the optimization problem $\bar{L}^G$, they need to
satisfy
\begin{align}
(\bbk^*_{X|V}+\bbsigma_Y)^{-1}+\bbm_U+\bbm_D&=(\bbk_{X|V}^{*})^{-1}\label{KKT_1}\\
(\bbk^*_{X|U}+\bbsigma_Z)^{-1}+\bbm_X&=(\bbk_{X|U}^{*}+\bbsigma_Y)^{-1}+\bbm_U\label{KKT_2}\\
\bbm_U(\bbk_{X|U}^*-\bbk_{X|V}^*)&=(\bbk_{X|U}^*-\bbk_{X|V}^*)\bbm_U=\bzero\label{KKT_3}\\
\bbm_D(\bbf(\bbd)-\bbk^*_{X|V})&=(\bbf(\bbd)-\bbk_{X|V}^*)\bbm_D=\bzero\label{KKT_4}\\
\bbm_X(\bbk_{X}-\bbk^*_{X|U})&=(\bbk_X-\bbk_{X|U}^*)\bbm_X=\bzero
\label{KKT_5}
\end{align}
for some positive semi-definite matrices $\bbm_U,\bbm_D,\bbm_X$.
\end{Lem}
The proof of Lemma~\ref{lemma_KKT} is given in
Appendix~\ref{proof_of_lemma_KKT}.

Next, we use channel enhancement~\cite{Shamai_MIMO}. In
particular, we enhance the legitimate user's side information as
follows.
\begin{align}
(\bbk_{X|U}^*+\tilde{\bbsigma}_Y)^{-1}=(\bbk_{X|U}^*+\bbsigma_Y)^{-1}+\bbm_U
\label{enhancement_def}
\end{align}
This new covariance matrix $\tilde{\bbsigma}_Y$ has some useful
properties which are listed in the following lemma.
\begin{Lem}
\label{lemma_enhancement} We have the following facts.
\begin{itemize}
\item $\bzero\preceq\tilde{\bbsigma}_Y$

\item $ \tilde{\bbsigma}_Y\preceq
\bbsigma_Y,\tilde{\bbsigma}_Y\preceq \bbsigma_Z$

\item
$(\bbk_{X|V}^*+\tilde{\bbsigma}_Y)^{-1}=(\bbk_{X|V}^*+\bbsigma_Y)^{-1}+\bbm_U$

\item
$(\bbk_{X|U}^{*}+\tilde{\bbsigma}_Y)^{-1}(\bbk_{X|V}^*+\tilde{\bbsigma}_Y)=
(\bbk_{X|U}^{*}+\bbsigma_Y)^{-1}(\bbk_{X|V}^*+\bbsigma_Y)$

\item
$(\bbk_{X|U}^{*}+\tilde{\bbsigma}_Y)^{-1}(\bbk_{X}+\tilde{\bbsigma}_Y)=
(\bbk_{X|U}^{*}+\bbsigma_Z)^{-1}(\bbk_{X}+\bbsigma_Z)$

\item
$(\bbk_{X|V}^*+\tilde{\bbsigma}_Y)^{-1}(\bbf(\bbd)+\tilde{\bbsigma}_Y)=(\bbk_{X|V}^*)^{-1}\bbf(\bbd)$
\end{itemize}
\end{Lem}
The proof of Lemma~\ref{lemma_enhancement} is given in
Appendix~\ref{proof_of_lemma_enhancement}. Using this new
covariance $\tilde{\bbsigma}_Y$, we define the {\it enhanced} side
information at the legitimate user $\tilde{\bby}$ as follows
\begin{align}
\tilde{\bby}&=\bbx+\tilde{\bbn}_Y
\label{enhanced_side_information}
\end{align}
where $\tilde{\bbn}_Y$ is a zero-mean Gaussian random vector with
covariance matrix $\tilde{\bbsigma}_Y$. Since we have
$\tilde{\bbsigma}_Y\preceq \bbsigma_Y$ and
$\tilde{\bbsigma}_Y\preceq \bbsigma_Z$ as stated in the second
statement of Lemma~\ref{lemma_enhancement}, without loss of
generality, we can assume that the following Markov chain exists.
\begin{align}
\bbx\rightarrow \tilde{\bby}\rightarrow \bby,\bbz
\label{mc_degraded}
\end{align}
Assuming that the Markov chain in (\ref{mc_degraded}) exists does
not incur any loss of generality because the rate, mutual
information leakage and distortion region $\mathcal{R}$ depends
only on the conditional marginal distributions
$p(\bby|\bbx),p(\bbz|\bbx)$ but not on the conditional joint
distribution $p(\bby,\bbz|\bbx)$. Now, we define the following
optimization problem:
\begin{align}
\bar{L}=\min_{\substack{U\rightarrow V\rightarrow \bbx\rightarrow
\tilde{\bby}\rightarrow \bby,\bbz\\ \bbk_{X|VY}\preceq \bbd}}
~I(V;\bbx)-I(V;\tilde{\bby}|U)+I(\bbx;\bbz|U) \label{bar_L}
\end{align}
We note that we have $I(V;\bby|U)\leq I(V;\tilde{\bby}|U)$ due to
the Markov chain in (\ref{mc_degraded}), which leads to the
following fact:
\begin{align}
L^G=\bar{L}^G\geq L(0,1)\geq \bar{L} \label{ineq_programs}
\end{align}
Moreover, unlike the original optimization problem $L(0,1)$ in
(\ref{Ie_min_program}), we can find the minimizer of the new
optimization problem $\bar{L}$ explicitly, as stated in the
following lemma.
\begin{Lem}
\label{lemma_lower_bound}
\begin{align}
\bar{L}=\frac{1}{2}\log\frac{|\bbk_X|}{|\bbf(\bbd)|}-\frac{1}{2}\log\frac{|\bbk_X+\tilde{\bbsigma}_Y|}{|\bbf(\bbd)+\tilde{\bbsigma}_Y|}
+\frac{1}{2}\log \frac{|\bbk_X+\bbsigma_Z|}{|\bbsigma_Z|}
\end{align}
\end{Lem}
We note that Lemma~\ref{lemma_lower_bound} implies that $U=\phi$
and a Gaussian $V$ leading to $\bbk_{X|V}=\bbf(\bbd)$ is the
minimizer of the optimization problem $\bar{L}$. The proof of
Lemma~\ref{lemma_lower_bound} is given in
Appendix~\ref{proof_of_lemma_lower_bound}.

Next, we show that indeed $L^G=\bar{L}^G=\bar{L}$ which, in view
of (\ref{ineq_programs}), will imply
$L(0,1)=\bar{L}=\bar{L}^G=L^G$. To this end, using
Lemma~\ref{lemma_lower_bound}, we have
\begin{align}
\bar{L}&=\frac{1}{2}\log\frac{|\bbk_X|}{|\bbf(\bbd)|}-\frac{1}{2}\log\frac{|\bbk_X+\tilde{\bbsigma}_Y|}{|\bbf(\bbd)+\tilde{\bbsigma}_Y|}
+\frac{1}{2}\log\frac{|\bbk_X+\bbsigma_Z|}{|\bbsigma_Z|}\\
&=\frac{1}{2}\log\frac{|\bbk_X|}{|\bbk_{X|V}^*|}-\frac{1}{2}\log\frac{|\bbk_X+\tilde{\bbsigma}_Y|}{|\bbk_{X|V}^*+\tilde{\bbsigma}_Y|}
+\frac{1}{2}\log\frac{|\bbk_X+\bbsigma_Z|}{|\bbsigma_Z|}
\label{lemma_enhancement_implies_1}\\
&=\frac{1}{2}\log\frac{|\bbk_X|}{|\bbk_{X|V}^*|}-\frac{1}{2}\log\frac{|\bbk_{X|U}^*+\tilde{\bbsigma}_Y|}{|\bbk_{X|V}^*+\tilde{\bbsigma}_Y|}
+\frac{1}{2}\log \frac{|\bbk_{X|U}^*+\bbsigma_Z|}{|\bbsigma_Z|}
\label{lemma_enhancement_implies_2}\\
&=\frac{1}{2}\log\frac{|\bbk_X|}{|\bbk_{X|V}^*|}-\frac{1}{2}\log\frac{|\bbk_{X|U}^*+\bbsigma_Y|}{|\bbk_{X|V}^*+\bbsigma_Y|}
+\frac{1}{2}\log \frac{|\bbk_{X|U}^*+\bbsigma_Z|}{|\bbsigma_Z|}
\label{lemma_enhancement_implies_3}\\
&=\bar{L}^G=L^G \label{final_step_xx}
\end{align}
where (\ref{lemma_enhancement_implies_1}) comes from the last
statement of Lemma~\ref{lemma_enhancement},
(\ref{lemma_enhancement_implies_2}) follows from the fifth
statement of Lemma~\ref{lemma_enhancement}, and
(\ref{lemma_enhancement_implies_3}) comes from the fourth
statement of Lemma~\ref{lemma_enhancement}. In view of
(\ref{ineq_programs}), (\ref{final_step_xx}) implies that
$L(0,1)=L^G$; completing the proof of Theorem~\ref{theorem_outer}
as well as the proof of Theorem~\ref{theorem_Ie_min} due to the
fact that $I_e^{\min}=L(0,1)$.

\section{General Case}
\label{sec:general_case}

We now consider the general case where the side information are
given by
\begin{align}
\bby&=\bbh_Y\bbx+\bbn_Y \label{general_1}\\
\bbz&=\bbh_Z \bbx+\bbn_Z \label{general_2}
\end{align}
where without loss of generality, we can assume that the
covariance matrices of Gaussian vectors $\bbn_Y$ and $\bbn_Z$ are
given by identity matrices. We denote the singular value
decomposition of $\bbh_Y$ and $\bbh_Z$ by $\bbh_Y=\bbq_Y
\bblambda_Y\bbr_Y^\top$ and $\bbh_Z=\bbq_Z\bblambda_Z
\bbr_Z^\top$, respectively. Since any invertible transformation
applied to the side information does not change the rate,
information leakage, and distortion region, the side information
given by (\ref{general_1})-(\ref{general_2}) and the side
information obtained by multiplying
(\ref{general_1})-(\ref{general_2}) by $\bbq_Y^\top,\bbq_Z^\top$,
respectively, yield the same rate, information leakage and
distortion region. In other words, the side information given by
(\ref{general_1})-(\ref{general_2}) and the side information given
by
\begin{align}
\bar{\bby}&=\bblambda_Y\bbr_Y^\top\bbx+\bar{\bbn}_Y \label{general_SVD_1}\\
\bar{\bbz}&=\bblambda_Z\bbr_Z^\top \bbx+\bar{\bbn}_Z
\label{general_SVD_2}
\end{align}
yield the same rate, information leakage and distortion region,
where the covariance matrices of $\bar{\bbn}_Y,\bar{\bbn}_Z$ are
given by identity matrices. Next, we claim that there is no loss
of generality to assume that the side information $\bar{\bby}$ and
$\bar{\bbz}$ have the same length as the source $\bbx$. To this
end, assume that the length of $\bar{\bby}$ is smaller than the
length of $\bbx$. In this case, simply, we can concatenate
$\bar{\bby}$ with some zero vector to ensure that both
$\bar{\bby}$ and $\bbx$ have the same length. Next, assume that
the length of $\bar{\bby}$ is larger than the length of $\bbx$. In
this case, $\bblambda_Y$ will definitely have at least ${\rm
length}(\bar{\bby})-{\rm length}(\bbx)$ diagonal elements which
are zero, and hence the corresponding entries in $\bar{\bby}$ will
come from only the noise. Since noise components are independent,
dropping these elements of $\bar{\bby}$ does not change the rate,
information leakage and distortion region. Thus, without loss of
generality, we can assume that ${\rm length}(\bar{\bby})={\rm
length}(\bbx)$, and hence without loss of generality, we can
assume that $\bblambda_Y$ is a square matrix. The same argument
applies to the eavesdropper's side information, and hence, without
loss of generality, we can also assume that $\bblambda_Z$ is a
square matrix. Next, we define the following side information
\begin{align}
\bar{\bby}_\alpha&=(\bblambda_Y+\alpha \bbi)\bbr_Y^\top\bbx+\bar{\bbn}_Y \label{general_alpha_1}\\
\bar{\bbz}_\alpha&=(\bblambda_Z+\alpha \bbi)\bbr_Z^\top
\bbx+\bar{\bbn}_Z \label{general_alpha_2}
\end{align}
where $\alpha >0$. We note that $(\bblambda_Y+\alpha \bbi)$ and
$(\bblambda_Y+\alpha \bbi)$ are invertible matrices. Since
multiplying the side information in
(\ref{general_alpha_1})-(\ref{general_2}) by some invertible
matrices does not change the rate, information leakage and
distortion region, the side information in
(\ref{general_alpha_1})-(\ref{general_alpha_2}) and the following
side information
\begin{align}
\bar{\bar{\bby}}_\alpha&=\bbx+\bar{\bbn}_{Y,\alpha} \label{general_alpha_1_aligned}\\
\bar{\bar{\bbz}}_\alpha&= \bbx+\bar{\bbn}_{Z,\alpha}
\label{general_alpha_2_aligned}
\end{align}
have the same rate, information leakage and distortion region,
where the covariance matrices of $\bar{\bbn}_{Y,\alpha}$ and
$\bar{\bbn}_{Z,\alpha}$ are given by
\begin{align}
\bbsigma_{Y,\alpha}&=\bbr_Y(\bblambda_Y+\alpha \bbi)^{-2} \bbr_Y^\top \label{covariance_y_alpha}\\
\bbsigma_{Z,\alpha}&=\bbr_Z(\bblambda_Z+\alpha \bbi)^{-2}
\label{covariance_z_alpha} \bbr_Z^\top
\end{align}
respectively. For a given distortion constraint $\bbd$, we denote
the rate and information leakage region for the side information
model given in (\ref{general_1})-(\ref{general_2}) by
$\mathcal{R}_o(\bbd)$, where the subscript $o$ stands for the
``original system'', and for the side information model given in
(\ref{general_alpha_1_aligned})-(\ref{general_alpha_2_aligned}) by
$\mathcal{R}_\alpha (\bbd)$. We have the following relationship
between $\mathcal{R}_o(\bbd)$ and $\mathcal{R}_\alpha (\bbd)$.
\begin{Lem}
\label{lemma_inclusion}
\begin{align}
\mathcal{R}_o(\bbd) \subseteq \lim_{\alpha \rightarrow
0}\mathcal{R}_\alpha(\bbd)
\end{align}
\end{Lem}
The proof of Lemma~\ref{lemma_inclusion} is given in
Appendix~\ref{proof_of_lemma_inclusion}. Next, using
Theorem~\ref{theorem_outer}, we obtain an outer bound for the
region $\lim_{\alpha\rightarrow 0}\mathcal{R}_\alpha (\bbd)$,
where this outer bound also serves as an outer bound for the
region $\mathcal{R}_o(\bbd)$ due to Lemma~\ref{lemma_inclusion}.
The corresponding result is stated in the following theorem.
\begin{Theo}
\label{theorem_outer_general}If $\bbd\preceq \bbk_{X|Y}$, any
$(R,I_e)\in\mathcal{R}_o(\bbd)$ satisfies
\begin{align}
R&\geq\frac{1}{2} \log
\frac{|\bbk_{X|Y}|}{|\bbd|}=\frac{1}{2}\log\frac{|\bbk_X|}{|\bbf_o(\bbd)|}-
\frac{1}{2}\log\frac{|\bbh_Y\bbk_X\bbh_Y^\top+\bbi|}{|\bbh_Y\bbf_o(\bbd)\bbh_Y^\top+\bbi|}\\
I_e&\geq \min_{\substack{\bzero\preceq \bbk_{X|V}\preceq
\bbk_{X|U}\preceq \bbk_X\\ \bbk_{X|V}\preceq \bbf_o(\bbd)}}
\frac{1}{2}\log\frac{|\bbk_X|}{|\bbk_{X|V}|}-\frac{1}{2}\log\frac{|\bbh_Y\bbk_{X|U}\bbh_Y^\top+\bbi|}{|\bbh_Y\bbk_{X|V}\bbh_Y^\top+\bbi|}
+\frac{1}{2} \log|\bbh_Y\bbk_{X|U}\bbh_Y^\top+\bbi|
\end{align}
where $\bbf_o(\bbd)=(\bbd^{-1}-\bbh_Y^\top\bbh_Y)^{-1}$.
\end{Theo}
The proof of Theorem~\ref{theorem_outer_general} is given in
Appendix~\ref{proof_of_theorem_outer_general}. We prove
Theorem~\ref{theorem_outer_general} in two steps. In the first
step, by using Theorem~\ref{theorem_outer}, we obtain an outer
bound for the region $\mathcal{R}_\alpha (\bbd)$, and in the
second step, we obtain the limit of this outer bound as $\alpha
\rightarrow 0$. As the outer bound in
Theorem~\ref{theorem_outer_general} basically comes from the outer
bound in Theorem~\ref{theorem_outer}, all our previous comments
and remarks about Theorem~\ref{theorem_outer} are also valid for
the outer bound in Theorem~\ref{theorem_outer_general}. Similar to
Theorem~\ref{theorem_outer}, Theorem~\ref{theorem_outer_general}
also provides the minimum information leakage to the eavesdropper
when the rate constraint on the transmitter is removed. Denoting
the corresponding minimum information leakage by $I_e^{\rm
min}(\bbd)$, we have the following theorem.
\begin{Theo}
\label{theorem_minimum_leakage_general}If $\bbd\preceq
\bbk_{X|Y}$, we have
\begin{align}
I_e^{\rm min}(\bbd)&\geq \min_{\substack{\bzero \preceq
\bbk_{X|V}\preceq \bbk_{X|U}\preceq \bbk_X\\ \bbk_{X|V}\preceq
\bbf_o(\bbd)}}
\frac{1}{2}\log\frac{|\bbk_X|}{|\bbk_{X|V}|}-\frac{1}{2}\log\frac{|\bbh_Y\bbk_{X|U}\bbh_Y^\top+\bbi|}{|\bbh_Y\bbk_{X|V}\bbh_Y^\top+\bbi|}
+\frac{1}{2} \log|\bbh_Y\bbk_{X|U}\bbh_Y^\top+\bbi|
\end{align}
where $\bbf_o(\bbd)=(\bbd^{-1}-\bbh_Y^\top\bbh_Y)^{-1}$.
\end{Theo}
As Theorem~\ref{theorem_minimum_leakage_general} basically comes
from Theorem~\ref{theorem_Ie_min}, all our previous comments and
remarks about Theorem~\ref{theorem_Ie_min} are also valid for
Theorem~\ref{theorem_minimum_leakage_general}.

\section{Conclusions}

In this paper, we study secure lossy source coding for vector
Gaussian sources, where the transmitter sends information about
the source in a way that the legitimate user can reconstruct the
source within a distortion level by using its side information.
Meanwhile, the transmitter wants to keep the mutual information
leakage to the eavesdropper to a minimum, where the eavesdropper
also has a side information about the source. We obtain an outer
bound for the achievable rate, mutual information leakage, and
distortion region. Moreover, we obtain the minimum mutual
information leakage to the eavesdropper when the legitimate user
needs to reconstruct the source within a certain distortion while
there is no constraint on the transmission rate.

\appendixpage
\appendices

\section{Proof of (\ref{I_e_inst_min_0})}
\label{proof_of_I_e_ins_0} We first define the following function
\begin{align}
R(D)=\min_{\substack{V\rightarrow X\rightarrow Y,Z\\
\sigma_{X|VY}^2\leq D}}~I(X;V,Z)
\end{align}
which is monotonically decreasing, continuous and convex in $D$.
Next, we note that when an instantaneous encoding scheme is used,
the minimum-mean-square-error estimator is given by
\begin{align}
\hat{X}_i&=E\left[X_i|g_{1}(X_1),\ldots,g_n(X_n),Y^n\right]\\
&=E\left[X_i|g_{i}(X_i),Y_i\right] \label{independence}
\end{align}
where (\ref{independence}) comes from the independence of
$(X_i,g_i(X_i),Y_i)$ across time. Consequently, when an
instantaneous encoding scheme is used, the
minimum-mean-square-error is given by
\begin{align}
\sigma_{X_i|g_i(X_i)Y_i}^2=E\left[\left(X_i-E\left[X_i|g_{i}(X_i),Y_i\right]\right)^2\right]
\end{align}
Assume that there exists an instantaneous encoding scheme that
achieves the distortion level $D$:
\begin{align}
\lim_{n\rightarrow \infty}\frac{1}{n}\sum_{i=1}^n
\sigma_{X_i|g_i(X_i)Y_i}^2\leq D \label{ach_distortion}
\end{align}
We now obtain a lower bound for the minimum information leakage
for this instantaneous encoding scheme as follows
\begin{align}
\lim_{n\rightarrow
\infty}\frac{1}{n}I(X^n;M,Z^n)&=\lim_{n\rightarrow
\infty}\frac{1}{n}I(X^n;g_1(X_1),\ldots,g_n(X_n),Z^n)\\
&=\lim_{n\rightarrow \infty}\frac{1}{n}\sum_{i=1}^n
I(X_i;g_i(X_i),Z_i) \label{independence_1}\\
&=\lim_{n\rightarrow \infty}\frac{1}{n}\sum_{i=1}^n
I(X_i;V_i,Z_i)\label{def_V}\\
&\geq \lim_{n\rightarrow \infty}\frac{1}{n}\sum_{i=1}^n
R\left(\sigma_{X_i|V_iY_i}^2\right) \label{def_R_D} \\
&\geq \lim_{n\rightarrow \infty}
R\left(\frac{1}{n}\sum_{i=1}^n\sigma_{X_i|V_iY_i}^2\right)\label{R_D_convex}\\
&= R\left(\lim_{n\rightarrow \infty}
\frac{1}{n}\sum_{i=1}^n\sigma_{X_i|V_iY_i}^2\right)\label{R_D_continuous}\\
&\geq R(D) \label{assumption_implies}
\end{align}
where (\ref{independence_1}) comes from the independence of
$(X_i,g_i(X_i),Z_i)$ across time, (\ref{def_V}) follows by setting
$V_i=g_{i}(X_i)$, (\ref{def_R_D}) comes from the definition of
$R(D)$, (\ref{R_D_convex}) is due to the convexity of $R(D)$ in
$D$, (\ref{R_D_continuous}) follows from the fact that $R(D)$ is
continuous in $D$, and (\ref{assumption_implies}) comes from
(\ref{ach_distortion}) and the fact that $R(D)$ is monotonically
decreasing in $D$.

\section{Proof of Lemma~\ref{lemma_strictly_nonzero}}
\label{proof_of_lemma_strictly_nonzero}

We first introduce two lemmas that will be used in the proof of
Lemma~\ref{lemma_strictly_nonzero}. Throughout this appendix, we
use notation $A\ci B$ to denote ``$A$ and $B$ are independent'' to
shorten the presentation.

\begin{Lem}
\label{lemma_dummy_1} Let $Q,T,W$ be arbitrary random variables.
If we have $Q\rightarrow T\rightarrow T+W$ and $T \ci W$. Then, we
have $(Q,T)\ci W$.
\end{Lem}
\begin{proof}
Since a set of random variables is independent iff their joint
characteristic function is the product of their individual
characteristic functions, to prove Lemma~\ref{lemma_dummy_1}, it
is sufficient to show the following.
\begin{align}
E\left[e^{s_1Q+s_2 T+s_3W}\right]=E\left[e^{s_1Q+s_2
T}\right]E\left[e^{s_3W}\right],\quad \forall (s_1,s_2,s_3)
\end{align}
We can show this as follows
\begin{align}
E\left[e^{s_1Q+s_2 T+s_3W}\right]&=E\left[E\left[e^{s_1Q+s_2
T+s_3W}\big|T\right]\right]\\
&=E\left[e^{(s_2-s_3)
T}~E\left[e^{s_1Q+s_3(T+W)}\big|T\right]\right]\\
&=E\left[e^{(s_2-s_3)
T}~E\left[e^{s_1Q}\big|T\right]E\left[e^{s_3(T+W)}\big|T\right]\right]\label{MC_implies_1}\\
&=E\left[e^{s_2
T}~E\left[e^{s_1Q}\big|T\right]E\left[e^{s_3W}\big|T\right]\right]\\
&=E\left[e^{s_2
T}~E\left[e^{s_1Q}\big|T\right]E\left[e^{s_3W}\right]\right]\label{independence_implies_1}\\
&=E\left[e^{s_2
T}~E\left[e^{s_1Q}\big|T\right]\right]E\left[e^{s_3W}\right]\label{independence_implies_2}\\
&=E\left[e^{s_1Q+s_2T}\right]E\left[e^{s_3W}\right]
\label{desired_result_1}
\end{align}
where (\ref{MC_implies_1}) comes from the Markov chain
$Q\rightarrow T \rightarrow T+W$ and
(\ref{independence_implies_1}) follows from the fact that $T\ci
W$. Equation (\ref{desired_result_1}) implies the independence
between $(Q,T)$ and $W$; completing the proof of
Lemma~\ref{lemma_dummy_1}.
\end{proof}

\begin{Lem}
\label{lemma_dummy_2} Let $Q,T,W$ be random variables satisfying
$(T,Q)\ci W$ and $Q\ci T+W$.  Then, we have $Q\ci T$.
\end{Lem}
\begin{proof}
Similar to the proof of Lemma~\ref{lemma_dummy_1}, here also we
use the fact that a set of random variables is independent iff
their joint characteristic function is the product of their
individual characteristic functions. To this end, since $(T,Q)\ci
W$, we have
\begin{align}
E\left[e^{s_1 W+s_2T+s_3Q}\right]&=E\left[e^{s_1
W}\right]E\left[e^{s_2T+s_3Q}\right],\quad \forall (s_1,s_2,s_3)
\label{independence_implies_3}
\end{align}
If we set $s_1=s_2$ in (\ref{independence_implies_3}), we get
\begin{align}
E\left[e^{s_2 W+s_2T+s_3Q}\right]&=E\left[e^{s_2
W}\right]E\left[e^{s_2T+s_3Q}\right],\quad \forall (s_2,s_3)
\label{independence_implies_4}
\end{align}
On the other hand, since $Q\ci T+W$, we have
\begin{align}
E\left[e^{s_2 W+s_2T+s_3Q}\right]&=E\left[e^{s_2(
W+T)}\right]E\left[e^{s_3Q}\right]\\
&=E\left[e^{s_2
W}\right]E\left[e^{s_2T}\right]E\left[e^{s_3Q}\right]\label{independence_implies_5}
\end{align}
where (\ref{independence_implies_5}) comes from the fact that
$T\ci W$. In view of (\ref{independence_implies_4}) and
(\ref{independence_implies_5}), we have
\begin{align}
E\left[e^{s_2T+s_3Q}\right]&=E\left[e^{s_2T}\right]E\left[e^{s_3Q}\right]
\end{align}
which implies that $T\ci Q$; completing the proof of
Lemma~\ref{lemma_dummy_2}.
\end{proof}

We now prove Lemma~\ref{lemma_strictly_nonzero}. We note that we
have $I(V;Y|Z)=0$ iff the Markov chain $V\rightarrow Z \rightarrow
Y$ holds. We prove by contradiction that when $D<\sigma^2_{x|y}$,
the Markov chain $V\rightarrow Z\rightarrow Y$ is not possible. To
this end, we note that the side information at the eavesdropper
can be written as
\begin{align}
Z&=X+N_y+\tilde{N}_z
\end{align}
or in other words, we have $N_z=N_y+\tilde{N}_z$ where
$\tilde{N}_z$ is a Gaussian random variable independent of
$(X,N_y)$ with variance $\sigma_z^2-\sigma_y^2>0$. Next, we note
that the Markov chain $V\rightarrow X\rightarrow Y\rightarrow Z$
implies $(V,X)\ci (N_y,\tilde{N}_z)$ in view of
Lemma~\ref{lemma_dummy_1}. Since $Y,Z$ are jointly Gaussian, $Y$
can be written as
\begin{align}
Y=\alpha Z+(Y-\alpha Z)
\end{align}
where $\alpha=E[YZ]/E[Z^2]$, and as a consequence of this $\alpha$
choice, we have $Z\ci Y-\alpha Z$. Hence, if we have the Markov
chain
\begin{align}
V\rightarrow Z \rightarrow Y=\alpha Z+(Y-\alpha Z)
\end{align}
then, Lemma~\ref{lemma_dummy_1} implies that $V\ci Y-\alpha Z$,
where $Y-\alpha Z$ is
\begin{align}
Y-\alpha Z= (1-\alpha)X+(1-\alpha) N_y-\tilde{N}_z
\end{align}
Since $(V,X)\ci(N_y,\tilde{N}_z)$, we have $(V,X)\ci (1-\alpha)
N_y-\tilde{N}_z$, and also $V\ci (1-\alpha)X+(1-\alpha)
N_y-\tilde{N}_z$ due to the assumption that the Markov chain
$V\rightarrow Z \rightarrow Y$ holds. Hence, in view of
Lemma~\ref{lemma_dummy_2}, we have $V\ci X$. Moreover, since we
have the Markov chain $V\rightarrow X \rightarrow Y$, $V\ci X$
implies that $V\ci (X,Y)$. Hence, if $V\ci (X,Y)$, we have
$\sigma_{x|vy}^2=\sigma_{x|y}^2$. However, if $D<\sigma_{x|y}^2$,
$V \ci X$ is not feasible, and this implies that the Markov chain
$V \rightarrow Z\rightarrow Y$ is not possible; completing the
proof of Lemma~\ref{lemma_strictly_nonzero}.

\section{Proof of (\ref{minimum_leakage_parallel})}
\label{proof_of_minimum_leakage_parallel}

Here, we provide the proof of (\ref{minimum_leakage_parallel}). To
this end, we consider a slightly more general case where the joint
distribution of the source and side information is given by
\begin{align}
p(\bx,\by,\bz)=\prod_{i=1}^L p(x_i,y_i,z_i) \label{parallel_pmf}
\end{align}
and the distortion constraint is imposed with a diagonal matrix
$\bbd$ whose diagonal entries are denoted by $D_1,\ldots,D_L$.
From Theorem~\ref{theorem_Gauss_single}, the minimum information
leakage is given by
\begin{align}
I_e^{\rm min}=\min_{\substack{U\rightarrow V\rightarrow
\bbx\rightarrow \bby,\bbz\\
\sigma_{X_i|VY^L}^2\leq
D_i,~i=1,\ldots,L}}~I(V;\bbx)-I(V;\bby|U)+I(\bbx;\bbz|U)
\label{minimum_parallel}
\end{align}
We first introduce the following auxiliary random variables
\begin{align}
U_i&=UY^{i-1}Z_{i+1}^L,\qquad i=1,\ldots,L \label{def_U}\\
V_i&=VY^{i-1}X_{i+1}^L,\qquad i=1,\ldots,L \label{def_V_again}
\end{align}
which satisfy the Markov chain
\begin{align}
U_i\rightarrow V_i\rightarrow X_i\rightarrow Y_i,Z_i
\label{MC_U_V}
\end{align}
which follows from (\ref{parallel_pmf}) and the Markov chain $
U\rightarrow V\rightarrow \bbx \rightarrow \bby,\bbz$.

Next, we introduce the following two lemmas.
\begin{Lem}{\bf(\!\!\cite[Lemma 7]{Korner})}
\label{lemma_CK_sum}Let $S^n,T^n$ be length-$n$ random vectors,
and $W$ be an arbitrary random variable. We have
\begin{align}
\sum_{i=1}^n I(T_{i+1}^n;S_i|WS^{i-1})=\sum_{i=1}^n
I(S^{i-1};T_i|WT_{i+1}^n)
\end{align}
\end{Lem}
Using Lemma~\ref{lemma_CK_sum}, the following lemma can be proved.
\begin{Lem}
\label{lemma_CK_sum_1}
\begin{align}
I(W;S^n)-I(W;T^n)=\sum_{i=1}^n
I(W;S_i|S^{i-1}T_{i+1}^n)-I(W;T_i|S^{i-1}T_{i+1}^n)
\end{align}
\end{Lem}
Now, we proceed with (\ref{minimum_parallel}) as follows
\begin{align}
\lefteqn{I_e^{\rm min}=\min_{\substack{U\rightarrow V\rightarrow
\bbx\rightarrow \bby,\bbz\\
\sigma_{X_i|VY^L}^2\leq
D_i,~i=1,\ldots,L}}~I(V;\bbx)-I(V;\bby|U)+I(\bbx;\bbz|U)}\\
&= \min_{\substack{U\rightarrow V\rightarrow
\bbx\rightarrow \bby,\bbz\\
\sigma_{X_i|VY^L}^2\leq
D_i,~i=1,\ldots,L}}~I(V;\bbx)-I(V;\bby)+I(U;\bby)-I(U;\bbz)+I(\bbx;\bbz)
\label{MC_generic_implies}\\
&=\min_{\substack{U\rightarrow V\rightarrow
\bbx\rightarrow \bby,\bbz\\
\sigma_{X_i|VY^L}^2\leq D_i,~i=1,\ldots,L}}~\sum_{i=1}^L
I(V;X_i|Y^{i-1},X_{i+1}^L)-I(V;Y_i|Y^{i-1},X_{i+1}^L)\nonumber\\
&\qquad\qquad \qquad \qquad \qquad +\sum_{i=1}^L
I(U;Y_i|Y^{i-1},Z_{i+1}^L)-I(U;Z_i|Y^{i-1},Z_{i+1}^L)+I(\bbx;\bbz)
\label{lemma_CK_sum_1_implies} \\
&=\min_{\substack{U\rightarrow V\rightarrow
\bbx\rightarrow \bby,\bbz\\
\sigma_{X_i|VY^L}^2\leq D_i,~i=1,\ldots,L}}~\sum_{i=1}^L
I(V;X_i|Y^{i-1},X_{i+1}^L)-I(V;Y_i|Y^{i-1},X_{i+1}^L)\nonumber\\
&\qquad\qquad \qquad \qquad \qquad +\sum_{i=1}^L
I(U;Y_i|Y^{i-1},Z_{i+1}^L)-I(U;Z_i|Y^{i-1},Z_{i+1}^L)+I(X_i;Z_i)
\label{parallel_pmf_implies}\\
&=\min_{\substack{U\rightarrow V\rightarrow
\bbx\rightarrow \bby,\bbz\\
\sigma_{X_i|VY^L}^2\leq D_i,~i=1,\ldots,L}}~\sum_{i=1}^L
I(Y^{i-1},X_{i+1}^L,V;X_i)-I(Y^{i-1},X_{i+1}^L,V;Y_i)\nonumber\\
&\qquad\qquad \qquad \qquad \qquad +\sum_{i=1}^L
I(Y^{i-1},Z_{i+1}^L,U;Y_i)-I(Y^{i-1},Z_{i+1}^L,U;Z_i)+I(X_i;Z_i)
\label{parallel_pmf_implies_1}\\
&=\min_{\substack{U\rightarrow V\rightarrow
\bbx\rightarrow \bby,\bbz\\
\sigma_{X_i|VY^L}^2\leq D_i,~i=1,\ldots,L}}~\sum_{i=1}^L
I(V_i;X_i)-I(V_i;Y_i)+
I(U_i;Y_i)-I(U_i;Z_i)+I(X_i;Z_i)\label{def_imply} \\
&=\min_{\substack{U\rightarrow V\rightarrow
\bbx\rightarrow \bby,\bbz\\
\sigma_{X_i|VY^L}^2\leq D_i,~i=1,\ldots,L}}~\sum_{i=1}^L
I(V_i;X_i)-I(V_i;Y_i|U_i)+ I(X_i;Z_i|U_i) \label{MC_U_V_implies}\\
&\geq \min_{\substack{U_i\rightarrow V_i\rightarrow
X_i\rightarrow Y_i,Z_i\\
\sigma_{X_i|V_iY_i}^2\leq D_i,~i=1,\ldots,L}}~\sum_{i=1}^L
I(V_i;X_i)-I(V_i;Y_i|U_i)+ I(X_i;Z_i|U_i) \label{cond_reduces}
\end{align}
where (\ref{MC_generic_implies}) comes from the Markov chain
$U\rightarrow V\rightarrow \bbx \rightarrow \bby,\bbz$,
(\ref{lemma_CK_sum_1_implies}) follows from
Lemma~\ref{lemma_CK_sum_1}, (\ref{parallel_pmf_implies}) and
(\ref{parallel_pmf_implies_1}) are due to (\ref{parallel_pmf}),
(\ref{def_imply}) follows from the definitions of $U_i,V_i$ in
(\ref{def_U}) and (\ref{def_V_again}), respectively,
(\ref{MC_U_V_implies}) comes from (\ref{MC_U_V}), and
(\ref{cond_reduces}) follows from
\begin{align}
\sigma^2_{X_i|VY^L}&\geq \sigma^2_{X_i|VY^LX_{i+1}^L} \label{cond_reduces_mmse}\\
&=\sigma^2_{X_i|VY^iX_{i+1}^L} \label{some_MC} \\
&=\sigma^2_{X_i|V_iY_i} \label{def_V_again_implies}
\end{align}
where (\ref{cond_reduces_mmse}) follows from the fact that
conditioning reduces MMSE (which will be shown in
Appendix~\ref{proof_of_conditioning_reduces_mmse}),
(\ref{some_MC}) comes from the following Markov chain
\begin{align}
X_i,V,Y^i\rightarrow X_{i+1}^L \rightarrow Y_{i+1}^L
\end{align}
which is a consequence of (\ref{parallel_pmf}) and the Markov
chain $U\rightarrow V\rightarrow \bbx\rightarrow\bby,\bbz$, and
(\ref{def_V_again_implies}) is obtained by using the definition of
$V_i$ given in (\ref{def_V_again}). Hence, (\ref{cond_reduces})
implies that when the joint distribution of the source and side
information can be factorized as in (\ref{parallel_pmf}), the
minimum information leakage is given by
\begin{align}
I_e^{\rm min}=\min_{\substack{U_i\rightarrow V_i\rightarrow
X_i\rightarrow Y_i,Z_i\\
\sigma_{X_i|V_iY_i}^2\leq D_i,~i=1,\ldots,L}}~\sum_{i=1}^L
I(V_i;X_i)-I(V_i;Y_i|U_i)+ I(X_i;Z_i|U_i)
\label{minimum_parallel_final}
\end{align}
We now specialize (\ref{minimum_parallel_final}) for the case
given in Example~\ref{example_double_aux}, where $L=2$ and we have
the following Markov chains
\begin{align}
X_1\rightarrow Y_1 \rightarrow Z_1 \label{sub-channel_1} \\
X_2\rightarrow Z_2 \rightarrow Y_2 \label{sub-channel_2}
\end{align}
Under these conditions, the minimum information leakage is given
by
\begin{align}
I_e^{\rm min}&=\min_{\substack{U_1\rightarrow V_1\rightarrow
X_1\rightarrow Y_1\rightarrow Z_1\\
\sigma_{X_1|V_1Y_1}^2\leq D_1}}~ I(V_1;X_1)-I(V_1;Y_1|U_1)+
I(X_1;Z_1|U_1)\nonumber\\
&\quad+\min_{\substack{U_2\rightarrow V_2\rightarrow
X_2\rightarrow Z_2 \rightarrow Y_2\\
\sigma_{X_2|V_2Y_2}^2\leq D_2}}~I(V_2;X_2)-I(V_2;Y_2|U_2)+
I(X_2;Z_2|U_2) \\
&=\min_{\substack{V_1\rightarrow
X_1\rightarrow Y_1\rightarrow Z_1\\
\sigma_{X_1|V_1Y_1}^2\leq D_1}}~ I(V_1;X_1)-I(V_1;Y_1)+
I(X_1;Z_1)\nonumber\\
&\quad +\min_{\substack{U_2\rightarrow V_2\rightarrow
X_2\rightarrow Z_2 \rightarrow Y_2\\
\sigma_{X_2|V_2Y_2}^2\leq D_2}}~I(V_2;X_2)-I(V_2;Y_2|U_2)+
I(X_2;Z_2|U_2) \label{sub-channel_1_implies}\\
&=\min_{\substack{V_1\rightarrow
X_1\rightarrow Y_1\rightarrow Z_1\\
\sigma_{X_1|V_1Y_1}^2\leq D_1}}~ I(V_1;X_1)-I(V_1;Y_1)+
I(X_1;Z_1)\nonumber\\
&\quad +\min_{\substack{ V_2\rightarrow
X_2\rightarrow Z_2 \rightarrow Y_2\\
\sigma_{X_2|V_2Y_2}^2\leq D_2}}~I(V_2;X_2)+ I(X_2;Z_2|V_2)
\label{sub-channel_2_implies}
\end{align}
where (\ref{sub-channel_1_implies})-(\ref{sub-channel_2_implies})
come from the following Markov chains
\begin{align}
U_1\rightarrow V_1 \rightarrow X_1 \rightarrow Y_1 \rightarrow
Z_1\\
U_2\rightarrow V_2 \rightarrow X_2 \rightarrow Z_2 \rightarrow Y_2
\end{align}
respectively; completing the proof.

\section{Proofs of (\ref{minimum_leakage_phi}) and (\ref{minimum_leakage_single})}
\label{proof_of_minimum_leakage_single}

We first prove (\ref{minimum_leakage_phi}). To this end, we note
that when the joint distribution of the source and side
information is given by
\begin{align}
p(\bx,\by,\bz)=\prod_{i=1}^L p(x_i,y_i,z_i)
\end{align}
and the distortion constraint is imposed by a diagonal matrix
$\bbd$ with diagonal entries $D_1,\ldots,D_L$, the minimum
information leakage is given by
\begin{align}
I_e^{\rm min}=\min_{\substack{U_i\rightarrow V_i\rightarrow
X_i\rightarrow Y_i,Z_i\\
\sigma_{X_i|V_iY_i}^2\leq D_i,~i=1,\ldots,L}}~\sum_{i=1}^L
I(V_i;X_i)-I(V_i;Y_i|U_i)+ I(X_i;Z_i|U_i)
\label{minimum_parallel_final_again}
\end{align}
as shown in Appendix~\ref{proof_of_minimum_leakage_parallel} (in
particular, see (\ref{minimum_parallel_final})). When we set
$U=\phi$, in other words, when we set $U_1=\phi,\ldots,U_L=\phi$,
(\ref{minimum_parallel_final_again}) reduces to
\begin{align}
I_e^{\rm min-\phi}=\min_{\substack{ V_i\rightarrow
X_i\rightarrow Y_i,Z_i\\
\sigma_{X_i|V_iY_i}^2\leq D_i,~i=1,\ldots,L}}~\sum_{i=1}^L
I(V_i;X_i)-I(V_i;Y_i)+ I(X_i;Z_i)
\label{minimum_parallel_final_again_1}
\end{align}
which is the desired result in (\ref{minimum_leakage_phi}).

Next, we prove (\ref{minimum_leakage_single}) by using
(\ref{minimum_parallel_final_again}). When we set $U=V$, in other
words, when we set $U_1=V_1,\ldots,U_L=V_L$ in
(\ref{minimum_parallel_final_again}), we get
\begin{align}
I_e^{\rm min}=\min_{\substack{U_i\rightarrow V_i\rightarrow
X_i\rightarrow Y_i,Z_i\\
\sigma_{X_i|V_iY_i}^2\leq D_i,~i=1,\ldots,L}}~\sum_{i=1}^L
I(V_i;X_i)+ I(X_i;Z_i|V_i)
\end{align}
which is the desired result in (\ref{minimum_leakage_single}).

\section{Proof of Lemma~\ref{lemma_equivalence}}
\label{proof_of_lemma_equivalence}

We note that since $\bbx,\bby$ are jointly Gaussian, we
have~\cite[page 155]{poor_book}
\begin{align}
\bbk_{X|Y}&=\bbk_X-\bbk_{XY}\bbk_Y^{-1}\bbk_{YX}\\
&=\bbk_X-\bbk_X(\bbk_X+\bbsigma_Y)^{-1}\bbk_X
\label{structure_1}\\
&=\bbk_{X}(\bbk_X+\bbsigma_Y)^{-1}\bbsigma_Y \label{future_use}
\end{align}
where (\ref{structure_1}) comes from the fact that
$\bby=\bbx+\bbn_Y$. Next, we have the following chain of
equalities
\begin{align}
\frac{|\bbk_X(\bbk_X+\bbsigma_Y)^{-1}|}{|\bbf(\bbd)(\bbf(\bbd)+\bbsigma_Y)^{-1}|}&=
\frac{|\bbk_X(\bbk_X+\bbsigma_Y)^{-1}\bbsigma_Y|}{|\bbf(\bbd)(\bbf(\bbd)+\bbsigma_Y)^{-1}\bbsigma_Y|}\\
&=
\frac{|\bbk_{X|Y}|}{|(\bbsigma_Y(\bbsigma_Y-\bbd)^{-1}\bbsigma_Y-\bbsigma_Y)\bbsigma_Y^{-1}(\bbsigma_Y-\bbd)|}\label{definition_FD}\\
&=\frac{|\bbk_{X|Y}|}{|\bbd|} \label{qed}
\end{align}
where (\ref{definition_FD}) follows from the definition of
$\bbf(\bbd)$, i.e., $
\bbf(\bbd)=\bbsigma_Y(\bbsigma_Y-\bbd)^{-1}\bbsigma_Y-\bbsigma_Y$.
Equation (\ref{qed}) implies (\ref{lemma_equivalence_statement});
completing the proof of Lemma~\ref{lemma_equivalence}.

\section{Proof of Lemma~\ref{lemma_equivalence_1}}
\label{proof_of_lemma_equivalence_1}

We first prove the first statement of the lemma. To this end,
using (\ref{future_use}), we have
\begin{align}
\bbk_{X|Y}
&=\bbk_{X}(\bbk_X+\bbsigma_Y)^{-1}\bbsigma_Y\\
&=\bbsigma_Y-\bbsigma_Y(\bbk_X+\bbsigma_Y)^{-1}\bbsigma_Y
\label{dummy_1}
\end{align}
Hence, using (\ref{dummy_1}), the constraint $\bbd \preceq
\bbk_{X|Y}$ can be expressed as
\begin{align}
\bbd\preceq
\bbsigma_Y-\bbsigma_Y(\bbk_X+\bbsigma_Y)^{-1}\bbsigma_Y
\end{align}
which is
\begin{align}
\bbsigma_Y(\bbk_X+\bbsigma_Y)^{-1}\bbsigma_Y \preceq
\bbsigma_Y-\bbd
\end{align}
where $\bbsigma_Y(\bbk_X+\bbsigma_Y)^{-1}\bbsigma_Y \succ\bzero $
implying $\bbsigma_Y-\bbd\succ \bzero$. Hence, $\bbsigma_Y-\bbd$
is non-singular, and $(\bbsigma_Y-\bbd)^{-1}$ exists.

Next, we prove the second statement of the lemma. To this end, we
note that since $(V,\bbx,\bby)$ are jointly Gaussian,
$\bby=\bbx+\bbn_Y$, and $V$ is independent of $\bbn_Y$,
$\bbk_{X|VY}$ is given by~\cite[page 155]{poor_book}
\begin{align}
\bbk_{X|VY}=\bbk_X-[~\bbk_{XV}~~\bbk_{X}~]~\bbm^{-1}~[~\bbk_{XV}~~
\bbk_{X}~]^\top \label{mmse_def}
\end{align}
where $\bbm$ is given by
\begin{align}
\bbm= \left[
\begin{array}{cc}
\bbk_V & \bbk_{VX}\\
\bbk_{XV} & \bbk_Y
\end{array}
\right]
\end{align}
Using block matrix inversion lemma~\cite[page
45]{matrix_cookbook}, $\bbm^{-1}$ can be obtained as
\begin{align}
\bbm^{-1}= \left[
\begin{array}{cc}
\bbk_V^{-1}+\bbk_V^{-1}\bbk_{VX}\bbdelta_M^{-1}\bbk_{XV}\bbk_V^{-1} & -\bbk_V^{-1}\bbk_{VX}\bbdelta_M^{-1}\\
-\bbdelta_M^{-1}\bbk_{XV}\bbk_V^{-1} &\bbdelta_M^{-1}
\end{array}
\right] \label{m_inverse}
\end{align}
where $\bbdelta_M$ is given by
\begin{align}
\bbdelta_M&=\bbk_Y-\bbk_{XV}\bbk_V^{-1}\bbk_{VX}\\
&=\bbk_X-\bbk_{XV}\bbk_V^{-1}\bbk_{VX}+\bbsigma_Y\\
&=\bbk_{X|V}+\bbsigma_Y \label{delta_def}
\end{align}
where the last equality follows from the fact that
$\bbk_{X|V}=\bbk_X-\bbk_{XV}\bbk_V^{-1}\bbk_{VX}$. Using
(\ref{m_inverse}) and (\ref{delta_def}), we get
\begin{align}
[~\bbk_{XV}~~\bbk_{X}~]~\bbm^{-1}=\big[~\bbsigma_Y\bbdelta_M^{-1}\bbk_{XV}\bbk_V^{-1}~~~\bbi-\bbsigma_Y\bbdelta_M^{-1}~\big]
\end{align}
using this in conjunction with (\ref{delta_def}), we obtain
\begin{align}
[~\bbk_{XV}~~\bbk_{X}~]~\bbm^{-1}~[~\bbk_{XV}~~\bbk_{X}~]^\top=\bbk_X-\bbsigma_Y+\bbsigma_Y\bbdelta_M^{-1}\bbsigma_Y
\label{complicated_mmse}
\end{align}
Using (\ref{complicated_mmse}) in (\ref{mmse_def}), we have
\begin{align}
\bbk_{X|VY}&=\bbsigma_Y-\bbsigma_Y\bbdelta_M^{-1}\bbsigma_Y\\
&=\bbsigma_Y-\bbsigma_Y(\bbk_{X|V}+\bbsigma_Y)^{-1}\bbsigma_Y
\label{complicated_mmse_1}
\end{align}
where (\ref{complicated_mmse_1}) follows from (\ref{delta_def}).
Thus, using (\ref{complicated_mmse_1}), the constraint
$\bbk_{X|VY}\preceq \bbd$ can be expressed as follows
\begin{align}
\bbsigma_Y-\bbsigma_Y(\bbk_{X|V}+\bbsigma_Y)^{-1}\bbsigma_Y
\preceq \bbd
\end{align}
from which, since $\bbsigma_Y-\bbd\succ \bzero $, the following
order can be obtained
\begin{align}
\bbk_{X|V} \preceq
\bbsigma_Y(\bbsigma_Y-\bbd)^{-1}\bbsigma_Y-\bbsigma_Y=\bbf(\bbd)
\end{align}
which completes the proof of Lemma~\ref{lemma_equivalence_1}.

\section{Conditioning Reduces MMSE}
\label{proof_of_conditioning_reduces_mmse}

Here, we prove that conditioning reduces MMSE. To this end, we
introduce the following lemma.

\begin{Lem}
\label{lemma_dummy} Let $\bbu$ and $\bbv$ be any two
$n$-dimensional random vectors and $g:\mathbb{R}^n\rightarrow
\mathbb{R}^n$. Then,
\begin{align}
E\left[g(\bbv)g^\top(\bbv)|\bbu=\bu\right]\succeq
E\left[g(\bbv)|\bbu=\bu\right]E\left[g^\top(\bbv)|\bbu=\bu\right]
\end{align}
\end{Lem}
\begin{proof}
The proof of this lemma comes from the following fact
\begin{align}
\bzero& \preceq
E\left[\left(g(\bbv)-E\left[g(\bbv)|\bbu=\bu\right]\right)\left(g(\bbv)-E\left[g(\bbv)|\bbu=\bu\right]\right)^\top|\bbu=\bu\right]\\
&=E\left[g(\bbv)g^\top(\bbv)|\bbu=\bu\right]-E\left[g(\bbv)|\bbu=\bu\right]E\left[g^\top(\bbv)|\bbu=\bu\right]
\end{align}
\end{proof}

We now prove the fact that conditioning reduces MMSE.
\begin{Lem}
\label{lemma_conditioning_reduces_mmse} If $U\rightarrow
V\rightarrow \bbx$, then $\bbk_{X|V}\preceq \bbk_{X|U}$.
\end{Lem}
\begin{proof}
We have
\begin{align}
\bbk_{X|V}&=E\left[\bbx\bbx^\top\right]-E\left[E\left[\bbx|\bbv\right]E\left[\bbx^\top|\bbv\right]\right]\\
&=E\left[\bbx\bbx^\top\right]-E\left[E\left[E\left[\bbx|\bbv\right]E\left[\bbx^\top|\bbv\right]|\bbu\right]\right]\\
&\preceq E\left[\bbx\bbx^\top\right]-E\left[
E\left[E\left[\bbx|\bbv\right]|\bbu\right]
E\left[E\left[\bbx^\top|\bbv\right]|\bbu\right] \right]
\label{lemma_dummy_implies} \\
&=
E\left[\bbx\bbx^\top\right]-E\left[E\left[\bbx|\bbu\right]E\left[\bbx^\top|\bbu\right]\right]
\label{markov_chain_has_smthing}
\end{align}
where (\ref{lemma_dummy_implies}) comes from
Lemma~\ref{lemma_dummy} and (\ref{markov_chain_has_smthing}) comes
from the following fact
\begin{align}
E\left[E\left[\bbx|\bbv\right]|\bbu\right]=E\left[\bbx|\bbu\right]
\end{align}
which is a consequence of the Markov chain $U\rightarrow
V\rightarrow \bbx$.
\end{proof}

\section{Proof of Lemma~\ref{lemma_equivalence_2}}
\label{proof_of_lemma_equivalence_2}

We now prove Lemma~\ref{lemma_equivalence_2}. Since any jointly
Gaussian $(U,V,\bbx)$ triple satisfying the Markov chain
$U\rightarrow V\rightarrow \bbx$ also satisfies $\bbk_{X|V}\preceq
\bbk_{X|U}$ due to Lemma~\ref{lemma_conditioning_reduces_mmse},
the feasible set of $\bar{L}^G$ already contains all jointly
Gaussian $(U,V)$ pairs satisfying the Markov chain $U\rightarrow
V\rightarrow \bbx$. Hence, we have $L^G\geq \bar{L}^G$. Next, we
show that $\bar{L}^G\geq L^G$ to complete the proof of
Lemma~\ref{lemma_equivalence_2}. To do so, we need to show that
for any jointly Gaussian $(U,V,\bbx)$ with conditional covariance
matrices $\bbk_{X|U}$ and $\bbk_{X|V}$ satisfying $\bzero\preceq
\bbk_{X|V}\preceq \bbk_{X|U}\preceq \bbk_X$ and $\bbk_{X|V}\preceq
\bbf(\bbd)$, there exists another jointly Gaussian $(U^G,V^G)$
pair such that this pair has the following properties
\begin{itemize}
\item $\bbk_{X|V^G}=\bbk_{X|V}$

\item $\bbk_{X|U^G}=\bbk_{X|U}$

\item $U^G\rightarrow V^G \rightarrow \bbx$

\end{itemize}
To this end, we note that $(U^G,V^G)$ can be represented as
\begin{align}
V^G&=\bba_V\bbx+\bbn_V \\
U^G&=\bba_U\bbx+\bbn_U
\end{align}
where $(\bbn_U,\bbn_V)$ and $\bbx$ are independent,
$\bbn_U,\bbn_V$ are zero-mean Gaussian random vectors with
identity covariance matrices. The cross covariance of $\bbn_U$ and
$\bbn_V$ is given by
$\bbsigma_{UV}=E\left[\bbn_U\bbn_V^\top\right]$, which needs to be
selected accordingly to ensure that $U^G\rightarrow V^G\rightarrow
\bbx$.

The conditional covariance $\bbk_{X|V^G}$ is given by~\cite[page
155]{poor_book}
\begin{align}
\bbk_{X|V^G}&=\bbk_X-\bbk_{XV^G}\bbk_{V^G}^{-1}\bbk_{V^GX}
\label{k_xv_g}
\end{align}
Since we are seeking a $V^G$ such that $\bbk_{X|V^G}=\bbk_{X|V}$,
we set $\bbk_{X|V^G}=\bbk_{X|V}$ in (\ref{k_xv_g}) yielding
\begin{align}
\bbk_{X|V}&=\bbk_X-\bbk_{XV^G}\bbk_{V^G}^{-1}\bbk_{V^GX}\\
&=\bbk_X-\bbk_{X}\bba_V^\top(\bba_V\bbk_X\bba_V^\top+\bbi)^{-1}\bba_V\bbk_{X}
\label{k_xv_g_1}
\end{align}
which is equivalent to
\begin{align}
\bbk_X^{-1}(\bbk_X-\bbk_{X|V})\bbk_X^{-1}=\bba_V^\top(\bba_V\bbk_X\bba_V^\top+\bbi)^{-1}\bba_V
\label{k_xv_g_2}
\end{align}
Next, we note the Woodbury matrix
identity~\cite{horn_johnson_book1}.
\begin{Lem}{\bf (\!\!\cite[page 17]{horn_johnson_book1})}
\label{lemma_woodbury}
\begin{align}
\left(\bba+\bbc\bbb\bbc^\top\right)^{-1}&=\bba^{-1}-\bba^{-1}\bbc\left(\bbb^{-1}+\bbc^{\top}\bba^{-1}\bbc\right)^{-1}\bbc^{\top}\bba^{-1}
\end{align}
\end{Lem}
Using Woodbury matrix identity, we get
\begin{align}
\left(\bba_V\bbk_X\bba_V^\top+\bbi\right)^{-1}&=
\bbi-\bba_V(\bbk_X^{-1}+\bba_V^\top\bba_V)^{-1}\bba_V^\top
\end{align}
using which in (\ref{k_xv_g_2}), we get
\begin{align}
\bbk_X^{-1}(\bbk_X-\bbk_{X|V})\bbk_X^{-1}&=\bba_V^\top\left[\bbi-\bba_V(\bbk_X^{-1}+\bba_V^\top\bba_V)^{-1}\bba_V^\top\right]\bba_V
\\
&=\bba_V^\top\bba_V-\bba_V^\top\bba_V(\bbk_X^{-1}+\bba_V^\top\bba_V)^{-1}\bba_V^\top
\bba_V\\
&=\bba_V^\top\bba_V-\bba_V^\top\bba_V(\bbk_X^{-1}+\bba_V^\top\bba_V)^{-1}\left(\bbk_X^{-1}+\bba_V^\top
\bba_V-\bbk_X^{-1}\right)\\
&=\bba_V^\top\bba_V(\bbk_X^{-1}+\bba_V^\top\bba_V)^{-1}
\bbk_X^{-1} \\
&=\left(\bbk_X^{-1}+\bba_V^\top\bba_V-\bbk_{X}^{-1}\right)(\bbk_X^{-1}+\bba_V^\top\bba_V)^{-1}
\bbk_X^{-1} \\
&=\bbk_X^{-1}-\bbk_{X}^{-1}(\bbk_X^{-1}+\bba_V^\top\bba_V)^{-1}
\bbk_X^{-1}
\end{align}
which implies
\begin{align}
\bbk_{X|V}=\left(\bbk_{X}^{-1}+\bba_V^\top\bba_V\right)^{-1}
\label{future_use_x}
\end{align}
which, in turn, implies
\begin{align}
\bba_V^\top\bba_V=\bbk_{X|V}^{-1}-\bbk_X^{-1}\label{condition_1}
\end{align}
Hence, if we select $\bba_V$ as satisfying (\ref{condition_1}), we
get $\bbk_{X|V^G}=\bbk_{X|V}$. Similarly, if we select $\bba_U$ to
satisfy
\begin{align}
\bba_U^\top\bba_U=\bbk_{X|U}^{-1}-\bbk_X^{-1}\label{condition_2}
\end{align}
then, we also have $\bbk_{X|U^G}=\bbk_{X|U}$.

Next, we will explicitly construct $\bba_{V}$ and $\bba_{U}$
matrices to satisfy (\ref{condition_1}) and (\ref{condition_2}),
respectively. To this end, we introduce the following lemma, which
will be used subsequently.
\begin{Lem}[\!\!\cite{simultanoues_diag}]
\label{lemma_simul_diag} Let $\bba,\bbb$ be two real symmetric
positive semi-definite matrices. Then, there exists a non-singular
matrix $\bbw$ such that
\begin{align}
\bba&=\bbw^\top \bblambda_A \bbw \\
\bbb&=\bbw^\top \bblambda_B \bbw \\
\end{align}
where $\bblambda_A$ and $\bblambda_B$ are diagonal matrices.
\end{Lem}
Lemma~\ref{lemma_simul_diag} states that two real symmetric
positive semi-definite matrices can be diagonalized
simultaneously. Using this fact in
(\ref{condition_1})-(\ref{condition_2}), we get
\begin{align}
\bbk_{X|V}^{-1}-\bbk_X^{-1}&=\bbw^\top\bblambda_V^2\bbw
\label{decompose_1}\\
\bbk_{X|U}^{-1}-\bbk_X^{-1}&=\bbw^\top\bblambda_U^2\bbw
\label{decompose_2}
\end{align}
for some non-singular matrix $\bbw$, and diagonal matrices
$\bblambda_U,\bblambda_V$. Since $\bbk_{X|V}\preceq \bbk_{X|U},$
we have $\bbk_{X|V}^{-1}\succeq \bbk_{X|U}^{-1}$, which, in view
of (\ref{decompose_1})-(\ref{decompose_2}) imply
\begin{align}
\bbw^\top\left(\bblambda_V^2-\bblambda_U^2\right)\bbw\succeq
\bzero \label{decomp_implies_1}
\end{align}
Since $\bbw$ is non-singular, (\ref{decomp_implies_1}) implies
that
\begin{align}
\bblambda_V\succeq \bblambda_U \label{decomp_implies_2}
\end{align}
Finally, we choose
\begin{align}
\bba_V&=\bblambda_V\bbw \\
\bba_U&=\bblambda_U\bbw
\end{align}
which, in view of (\ref{condition_1})-(\ref{condition_2}) and
(\ref{decompose_1})-(\ref{decompose_2}), imply
$\bbk_{X|V^G}=\bbk_{X|V}$ and $\bbk_{X|U^G}=\bbk_{X|U}$.

Next, we show that a proper selection the cross-covariance matrix
$\bbsigma_{UV}$ would yield the desired Markov chain
$U^{G}\rightarrow V^{G} \rightarrow \bbx$. To this end, we
introduce the following matrix
\begin{align}
\bba_{UV}=\bblambda_U\bblambda_V^{\dag}
\end{align}
where the diagonal matrix $\bblambda_V^{\dag}$ is defined as
follows:
\begin{align}
\Lambda_{V,ii}^{\dag}= \left\{
\begin{array}{cc}
\frac{1}{\Lambda_{V,ii}},&\quad {\rm if}~~\Lambda_{V,ii}\neq 0\\
0,&\quad {\rm otherwise}
\end{array}
 \right.
\end{align}
Since $\bblambda_U\preceq \bblambda_V$, we have $\bblambda_U
\bblambda_V^{\dag} \bblambda_V=\bblambda_U$. Hence, we have
\begin{align}
\bba_{UV}\bba_{V}=\bba_{U}\label{towards_3rd_cond}
\end{align}
We also note the following
\begin{align}
\bba_{UV}\bba_{UV}^\top=\bblambda_U\big(\bblambda_V^{\dag}\big)^{2}\bblambda_U
\preceq \bbi
\end{align}
since $\bblambda_U\preceq \bblambda_V$.

Now, we are ready to show that $U^G$ and $V^G$ satisfy the Markov
chain $U^G\rightarrow V^G\rightarrow \bbx$ by specifying
$\bbsigma_{UV}$. We set $\bbn_U$ as follows
\begin{align}
\bbn_U&=\bba_{UV}\bbn_V+\tilde{\bbn} \label{covariance_structure}
\end{align}
where $\tilde{\bbn}$ is a zero-mean Gaussian random vector with
covariance matrix $\bbi-\bba_{UV}\bba_{UV}^\top$, and is
independent of $\bbn_V$. In view of (\ref{covariance_structure}),
we have
\begin{align}
U^G&=\bba_U\bbx+\bbn_U\\
&=\bba_{UV}\bba_V\bbx+\bba_{UV}\bbn_V+\tilde{\bbn} \\
&=\bba_{UV}V^G+\tilde{\bbn}
\end{align}
which implies that $(U^G,V^G)$ satisfy the Markov chain
$U^G\rightarrow V^G\rightarrow \bbx$; completing the proof.

\section{Proof of Lemma~\ref{lemma_KKT}}
\label{proof_of_lemma_KKT}

The Lagrangian for the optimization problem $\bar{L}^G$ is given
as follows
\begin{align}
\mathcal{L}\left(\bar{L}^G\right)&=\frac{1}{2}\log\frac{|\bbk_X|}{|\bbk_{X|V}|}
-\frac{1}{2}\log\frac{|\bbk_{X|U}+\bbsigma_Y|}{|\bbk_{X|V}+\bbsigma_Y|}
+\frac{1}{2}\log\frac{|\bbk_{X|U}+\bbsigma_Z|}{|\bbsigma_Z|} -{\rm
tr}(\bbm_0\bbk_{X|V})\nonumber\\
&\quad -{\rm tr}(\bbm_U(\bbk_{X|U}-\bbk_{X|V})) -{\rm
tr}(\bbm_X(\bbk_X-\bbk_{X|U}))-{\rm
tr}(\bbm_D(\bbf(\bbd)-\bbk_{X|V})) \label{Lagrangian}
\end{align}
where the positive semi-definite matrices
$\bbm_0,\bbm_U,\bbm_D,\bbm_X$ are the Lagrange multipliers for the
following constraints
\begin{align}
\bbk_{X|V}&\succeq \bzero \\
\bbk_{X|U}-\bbk_{X|V}&\succeq \bzero \\
\bbf(\bbd)-\bbk_{X|V}&\succeq \bzero \\
\bbk_X-\bbk_{X|U}&\succeq \bzero
\end{align}
respectively. Let $\bbk_{X|V}^*$ and $\bbk_{X|U}^*$ be the
minimizers of the optimization problem $\bar{L}^G$. Using
(\ref{Lagrangian}), the KKT conditions can be found as follows.
\begin{align}
\nabla_{\bbk_{X|V}}\mathcal{L}(\bar{L}^G)\mid_{\bbk_{X|V}=\bbk_{X|V}^*}&=\bzero \label{pre_KKT_1}\\
\nabla_{\bbk_{X|U}}\mathcal{L}(\bar{L}^G)\mid_{\bbk_{X|U}=\bbk_{X|U}^*}&=\bzero\label{pre_KKT_2}\\
{\rm tr}(\bbm_0\bbk_{X|V}^*)&=0\label{pre_KKT_3}\\
{\rm tr}(\bbm_U(\bbk_{X|U}^*-\bbk_{X|V}^*))&=0\label{pre_KKT_4}\\
{\rm tr}(\bbm_D(\bbf(\bbd)-\bbk_{X|V}^*))&=0 \label{pre_KKT_5}\\
{\rm tr}(\bbm_X(\bbk_X-\bbk_{X|U}^*))&=0\label{pre_KKT_6}
\end{align}
We first note that we have $\bbk_{X|V}^{*}\succ \bzero$, otherwise
$\bar{L}^G\rightarrow \infty$. Hence, using the fact that if
$\bba\succeq \bzero,\bbb\succeq \bzero $, ${\rm tr}(\bba\bbb)\geq
0$, and (\ref{pre_KKT_3}), we get $\bbm_0=\bzero$. Next, using the
fact that $\bbm_0=\bzero$ in (\ref{pre_KKT_1}), we get the KKT
condition given in (\ref{KKT_1}). Equation (\ref{pre_KKT_2})
implies (\ref{KKT_2}). Finally, using the fact that $\bba\succeq
\bzero ,\bbb\succeq \bzero $, ${\rm tr}(\bba\bbb)={\rm
tr}(\bbb\bba)\geq 0$ in (\ref{pre_KKT_4})-(\ref{pre_KKT_6}), we
can get the KKT conditions given in (\ref{KKT_3})-(\ref{KKT_5}),
respectively.

\section{Proof of Lemma~\ref{lemma_enhancement}}
\label{proof_of_lemma_enhancement}

We start with the second statement of the lemma. To this end, we
note that (\ref{KKT_2}) and (\ref{enhancement_def}) imply the
following.
\begin{align}
(\bbk_{X|U}^*+\tilde{\bbsigma}_Y)^{-1}&=(\bbk_{X|U}^*+\bbsigma_Y)^{-1}+\bbm_U \label{enh_implies_0}\\
&=(\bbk_{X|U}^*+\bbsigma_Z)^{-1}+\bbm_X \label{enh_implies}
\end{align}
Next, using the fact that if $\bba\succ \bzero,\bbb\succ \bzero $
and $\bba \succeq \bbb$, we have $\bba^{-1}\preceq \bbb^{-1}$ in
conjunction with the fact that $\bbm_U\succeq \bzero,\bbm_X\succeq
\bzero $, we can obtain the second statement of the lemma from
(\ref{enh_implies_0})-(\ref{enh_implies}).

Next, we consider the third statement of the lemma as follows
\begin{align}
\lefteqn{\bbk_{X|V}^*+\tilde{\bbsigma}_Y}\nonumber\\
&=\bbk_{X|V}^*+\left[(\bbk_{X|U}^*+\bbsigma_Y)^{-1}+\bbm_U\right]^{-1}-\bbk_{X|U}^*
\label{enh_implies_0_implies} \\
&=\bbk_{X|V}^*+\left[\bbi+(\bbk_{X|U}^*+\bbsigma_Y)\bbm_U\right]^{-1}(\bbk_{X|U}^*+\bbsigma_Y)-\bbk_{X|U}^*\\
&=\bbk_{X|V}^*+\left[\bbi+(\bbk_{X|U}^*-\bbk_{X|V}^*+\bbk_{X|V}^*+\bbsigma_Y)\bbm_U\right]^{-1}(\bbk_{X|U}^*+\bbsigma_Y)-\bbk_{X|U}^*\\
&=\bbk_{X|V}^*+\left[\bbi+(\bbk_{X|V}^*+\bbsigma_Y)\bbm_U\right]^{-1}(\bbk_{X|U}^*+\bbsigma_Y)-\bbk_{X|U}^* \label{KKT_3_implies}\\
&=\bbk_{X|V}^*+\left[(\bbk_{X|V}^*+\bbsigma_Y)^{-1}+\bbm_U\right]^{-1}(\bbk_{X|V}^*+\bbsigma_Y)^{-1}(\bbk_{X|U}^*+\bbsigma_Y)-\bbk_{X|U}^*
\\
&=\bbk_{X|V}^*+\left[(\bbk_{X|V}^*+\bbsigma_Y)^{-1}+\bbm_U\right]^{-1}(\bbk_{X|V}^*+\bbsigma_Y)^{-1}(\bbk_{X|U}^*-\bbk_{X|V}^*+\bbk_{X|V}^*+\bbsigma_Y)
\nonumber\\
&\quad -\bbk_{X|U}^*
\\
&=\bbk_{X|V}^*+\left[(\bbk_{X|V}^*+\bbsigma_Y)^{-1}+\bbm_U\right]^{-1}(\bbk_{X|V}^*+\bbsigma_Y)^{-1}(\bbk_{X|U}^*-\bbk_{X|V}^*)
\nonumber\\
&\quad +\left[(\bbk_{X|V}^*+\bbsigma_Y)^{-1}+\bbm_U\right]^{-1}
-\bbk_{X|U}^*
\\
&=\bbk_{X|V}^*+\left[(\bbk_{X|V}^*+\bbsigma_Y)^{-1}+\bbm_U\right]^{-1}\left[(\bbk_{X|V}^*+\bbsigma_Y)^{-1}+\bbm_U\right](\bbk_{X|U}^*-\bbk_{X|V}^*)
\nonumber\\
&\quad +\left[(\bbk_{X|V}^*+\bbsigma_Y)^{-1}+\bbm_U\right]^{-1}
-\bbk_{X|U}^* \label{KKT_3_implies_1}
\\
&=\bbk_{X|V}^*+(\bbk_{X|U}^*-\bbk_{X|V}^*)
+\left[(\bbk_{X|V}^*+\bbsigma_Y)^{-1}+\bbm_U\right]^{-1}
-\bbk_{X|U}^*
\\
&=\left[(\bbk_{X|V}^*+\bbsigma_Y)^{-1}+\bbm_U\right]^{-1}
\end{align}
where (\ref{enh_implies_0_implies}) comes from
(\ref{enh_implies_0}), (\ref{KKT_3_implies}) and
(\ref{KKT_3_implies_1}) follow from (\ref{KKT_3}).

Now, we consider the fourth statement of the lemma as follows
\begin{align}
(\bbk_{X|U}^{*}+\tilde{\bbsigma}_Y)^{-1}(\bbk_{X|V}^*+\tilde{\bbsigma}_Y)&=\bbi+(\bbk_{X|U}^{*}+\tilde{\bbsigma}_Y)^{-1}(\bbk_{X|V}^*-\bbk_{X|U}^*)\\
&=\bbi+\left[(\bbk_{X|U}^{*}+\bbsigma_Y)^{-1}+\bbm_U\right](\bbk_{X|V}^*-\bbk_{X|U}^*)\label{enh_implies_0_implies_1}\\
&=\bbi+(\bbk_{X|U}^{*}+\bbsigma_Y)^{-1}(\bbk_{X|V}^*-\bbk_{X|U}^*)\label{KKT_3_implies_2}\\
&=(\bbk_{X|U}^{*}+\bbsigma_Y)^{-1}(\bbk_{X|V}^*+\bbsigma_Y)
\end{align}
where (\ref{enh_implies_0_implies_1}) follows from
(\ref{enh_implies_0}), and (\ref{KKT_3_implies_2}) comes from
(\ref{KKT_3}).

Next, we consider the fifth statement of the lemma as follows
\begin{align}
(\bbk_{X|U}^*+\tilde{\bbsigma}_Y)^{-1}(\bbk_X+\tilde{\bbsigma}_Y)&=\bbi+(\bbk_{X|U}^*+\tilde{\bbsigma}_Y)^{-1}(\bbk_X-\bbk_{X|U}^*)\\
&=\bbi+\left[(\bbk_{X|U}^*+\bbsigma_Z)^{-1}+\bbm_X\right](\bbk_X-\bbk_{X|U}^*)\label{enh_implies_implies}\\
&=\bbi+(\bbk_{X|U}^*+\bbsigma_Z)^{-1}(\bbk_X-\bbk_{X|U}^*)\label{KKT_5_implies}\\
&=(\bbk_{X|U}^*+\bbsigma_Z)^{-1}(\bbk_X+\bbsigma_Z)
\end{align}
where (\ref{enh_implies_implies}) comes from (\ref{enh_implies}),
and (\ref{KKT_5_implies}) is due to (\ref{KKT_5}).

Now, we prove the last statement of the lemma. To this end, we
note that the third statement of this lemma and (\ref{KKT_1})
imply the following
\begin{align}
(\bbk_{X|V}^*+\tilde{\bbsigma}_Y)^{-1}+\bbm_D=(\bbk_{X|V}^*)^{-1}
\label{enh_implies_again}
\end{align}
which will be used in the sequel. Now, the last statement of this
lemma follows from
\begin{align}
(\bbk_{X|V}^*+\tilde{\bbsigma}_Y)^{-1}(\bbf(\bbd)+\tilde{\bbsigma}_Y)&=\bbi+(\bbk_{X|V}^*+\tilde{\bbsigma}_Y)^{-1}(\bbf(\bbd)-\bbk_{X|V}^{*})\\
&=\bbi+\left[(\bbk_{X|V}^*)^{-1}-\bbm_D\right](\bbf(\bbd)-\bbk_{X|V}^{*})\label{enh_implies_again_1}\\
&=\bbi+(\bbk_{X|V}^*)^{-1}(\bbf(\bbd)-\bbk_{X|V}^{*})
\label{KKT_4_implies}\\
&=(\bbk_{X|V}^*)^{-1}\bbf(\bbd)
\end{align}
where (\ref{enh_implies_again_1}) comes from
(\ref{enh_implies_again}), and (\ref{KKT_4_implies}) is due to
(\ref{KKT_4}).

Finally, we note that (\ref{enh_implies_again}) also implies the
first statement of the lemma; completing the proof.

\section{Proof of Lemma~\ref{lemma_lower_bound}}
\label{proof_of_lemma_lower_bound}

\subsection{Background}

\label{sec:background} We need some properties of the Fisher
information and the differential entropy, which are provided next.

\begin{Def}{\bf (\!\!\cite[Definition 3]{MIMO_BC_Secrecy})}
Let $(\bbu,\bbx)$ be an arbitrarily correlated length-$n$ random
vector pair with well-defined densities. The conditional Fisher
information matrix of $\bbx$ given $\bbu$ is defined as
\begin{align}
\bbj(\bbx|\bbu)=E\left[\brho(\bbx|\bbu)\brho(\bbx|\bbu)^\top\right]
\end{align}
where the expectation is over the joint density $f(\bu,\bx)$, and
the conditional score function $\brho(\bx|\bu)$ is
\begin{align}
\brho(\bx|\bu)&=\nabla \log f(\bx|\bu)=\left[~\frac{\partial\log
f(\bx|\bu)}{\partial x_1}~~\ldots~~\frac{\partial\log
f(\bx|\bu)}{\partial x_n}~\right]^\top
\end{align}
\end{Def}

We first present the conditional form of the Cramer-Rao
inequality, which is proved in~\cite{MIMO_BC_Secrecy}.
\begin{Lem}{\bf(\!\!\cite[Lemma 13]{MIMO_BC_Secrecy})}
\label{lemma_conditional_crb_vector} Let $\bbu,\bbx$ be
arbitrarily correlated random vectors with well-defined densities.
Let the conditional covariance matrix of $\bbx$ be ${\rm
Cov}(\bbx|\bbu)\succ \bzero$, then we have
\begin{align}
\bbj(\bbx|\bbu)\succeq {\rm Cov}(\bbx|\bbu)^{-1}
\end{align}
which is satisfied with equality if $(\bbu,\bbx)$ is jointly
Gaussian with conditional covariance matrix ${\rm
Cov}(\bbx|\bbu)$.
\end{Lem}

The following lemma will be used in the upcoming proof. The
unconditional version of this lemma, i.e., the case $\bbt=\phi$,
is proved in~\cite[Lemma 6]{MIMO_BC_Secrecy}.
\begin{Lem}{\bf(\!\!\cite[Lemma 6]{MIMO_BC_Secrecy})}
\label{lemma_change_in_fisher} Let $\bbt,\bbu,\bbv_1,\bbv_2$ be
random vectors such that $(\bbt,\bbu)$ and \break$(\bbv_1,\bbv_2)$
are independent. Moreover, let $\bbv_1,\bbv_2$ be Gaussian random
vectors with covariance matrices $\bbsigma_1,\bbsigma_2$ such that
$\bzero \prec \bbsigma_1 \preceq \bbsigma_2$. Then, we have
\begin{align}
\bbj^{-1}(\bbu+\bbv_2|\bbt)-\bbsigma_2 \succeq
\bbj^{-1}(\bbu+\bbv_1|\bbt)-\bbsigma_1
\end{align}
\end{Lem}

The following lemma will also be used in the upcoming proof.
\begin{Lem}{\bf(\!\!\cite[Lemma 8]{MIMO_BC_Secrecy})}
\label{Shamai_s_lemma} Let $\bbk_1,\bbk_2$ be positive
semi-definite matrices satisfying
$\bzero\preceq\bbk_1\preceq\bbk_2$, and $\mathbf{f}(\bbk)$ be a
matrix-valued function such that $\mathbf{f}(\bbk)\succeq\bzero$
for $\bbk_1\preceq\bbk\preceq \bbk_2$. Moreover,
$\mathbf{f}(\bbk)$ is assumed to be gradient of a scalar field.
Then, we have
\begin{align}
\int_{\bbk_1}^{\bbk_2}\mathbf{f}(\bbk)d\bbk \geq 0
\end{align}
\end{Lem}

The following generalization of the de Bruijn identity~\cite{Stam,
Blachman} is due to~\cite{Palomar_Gradient}, where the
unconditional form of this identity, i.e., $\bbu=\phi$, is proved.
Its generalization to this conditional form for an arbitrary
$\bbu$ is rather straightforward, and is given in~\cite[Lemma
16]{MIMO_BC_Secrecy}.
\begin{Lem}{\bf(\!\!\cite[Lemma 16]{MIMO_BC_Secrecy})}
\label{gradient_fisher_conditional} Let $(\bbu,\bbx)$ be an
arbitrarily correlated random vector pair with finite second order
moments, and also be independent of the random vector $\bbn$ which
is zero-mean Gaussian with covariance matrix
$\bbsigma_N\succ\bzero$. Then, we have
\begin{align}
\nabla_{\bbsigma_N} h(\bbx+\bbn|\bbu)=\frac{1}{2}
\bbj(\bbx+\bbn|\bbu)
\end{align}
\end{Lem}

The following lemma provides a connection between the conditional
covariance matrix and the Fisher information matrices of a random
vector.
\begin{Lem}
\label{lemma_connection} Let $(V,\bbx)$ be two arbitrary random
vectors with finite second moments, and $\bbn$ be a zero-mean
Gaussian random vector with covariance matrix $\bbsigma_N$. Let
$\bby=\bbx+\bbn$. Assume $(V,\bbx)$ and $\bbn$ are independent. We
have
\begin{align}
\bbk_{X|VY}=\bbsigma_N-\bbsigma_N\bbj(\bbx+\bbn|V)\bbsigma_N
\end{align}
\end{Lem}
Lemma~\ref{lemma_connection} is proved in~\cite{Palomar_Gradient}
for $V=\phi$. Its generalization to the current conditional form
can be obtained by using the conditional Fisher information and
Lemma~\ref{gradient_fisher_conditional}.

\subsection{Proof}

We first consider the cost function of the optimization problem
$\bar{L}$
\begin{align}
C(\bar{L})&=I(V;\bbx)-I(V;\tilde{\bby}|U)+I(\bbx;\bbz|U) \label{original}\\
&=I(V;\bbx)-I(V;\tilde{\bby})+I(U;\tilde{\bby})+I(\bbx;\bbz)-I(U;\bbz) \label{mc_degraded_implies_0}\\
&=I(V;\bbx)-I(V;\tilde{\bby})+I(U;\tilde{\bby},\bbz)+I(\bbx;\bbz)-I(U;\bbz)\label{mc_degraded_implies}\\
&=I(V;\bbx)-I(V;\tilde{\bby})+I(U;\tilde{\bby}|\bbz)+I(\bbx;\bbz)
\\
&\geq I(V;\bbx)-I(V;\tilde{\bby})+I(\bbx;\bbz)
\label{nonnegativity}
\end{align}
where (\ref{mc_degraded_implies_0})-(\ref{mc_degraded_implies})
come from the following Markov chain
\begin{align}
U\rightarrow V\rightarrow \bbx \rightarrow \tilde{\bby}
\rightarrow \bby,\bbz
\end{align}
and (\ref{nonnegativity}) comes from the non-negativity of the
mutual information. On the other hand, (\ref{nonnegativity}) can
be obtained from (\ref{bar_L}) by choosing $U=\phi$, i.e., we have
\begin{align}
\bar{L}\leq \min_{\substack{V\rightarrow \bbx \rightarrow
\tilde{\bby}\rightarrow \bby,\bbz\\\bbk_{X|VY}\preceq
\bbd}}~~I(V;\bbx)-I(V;\tilde{\bby})+I(\bbx;\bbz)
\label{nonnegativity_x}
\end{align}
Hence, (\ref{nonnegativity}) and (\ref{nonnegativity_x}) imply the
following
\begin{align}
\bar{L}&=\min_{\substack{V\rightarrow \bbx \rightarrow
\tilde{\bby}\rightarrow \bby,\bbz\\\bbk_{X|VY}\preceq
\bbd}}~~I(V;\bbx)-I(V;\tilde{\bby})+I(\bbx;\bbz)\\
& =\min_{\substack{V\rightarrow \bbx \rightarrow
\tilde{\bby}\rightarrow \bby,\bbz\\\bbk_{X|VY}\preceq
\bbd}}~~I(V;\bbx|\tilde{\bby})+I(\bbx;\bbz)\label{optimization_similar}
\end{align}
where (\ref{optimization_similar}) comes from the Markov chain
$V\rightarrow \bbx \rightarrow \tilde{\bby}$. We note that the
optimization problem in (\ref{optimization_similar}) is similar to
the one we already studied in
(\ref{future_use_0})-(\ref{monotonocity}). Indeed, if the
constraint $\bbk_{X|VY}\preceq \bbd$ in
(\ref{optimization_similar}) was $\bbk_{X|V\tilde{Y}}\preceq
\bbd$, both optimization problems would be identical, and using
the analysis in (\ref{future_use_0})-(\ref{monotonocity}), we
could conclude that (\ref{optimization_similar}) is minimized by a
Gaussian $V$ satisfying $\bbk_{X|V\tilde{Y}}\preceq \bbd$.
However, the difference between these two constraints necessitates
a new proof, and indeed, showing the optimality of Gaussian $V$
for the optimization problem in (\ref{optimization_similar}) is
not as straightforward as showing the optimality of Gaussian $V$
for the optimization problem in (\ref{future_use_0}).

We find the minimizer for the optimization problem $\bar{L}$ in
two steps. In the first step, for a given feasible $V$, we
explicitly construct a feasible Gaussian $\bar{V}$ which provides
the same value for the cost function of $\bar{L}$ as the original
$V$ does. Thus, this first step implies that restricting $V$ to be
Gaussian does not change the optimum value of the optimization
problem $\bar{L}$. Consequently, in the second step of the proof,
we minimize $\bar{L}$ over all feasible  Gaussian $V$. To this
end, we note that the cost function of the optimization problem
$\bar{L}$ can be written as
\begin{align}
C(\bar{L})=h(\tilde{\bby}|V)-h(\bbx|V)+c \label{cost_function}
\end{align}
for some constant $c$, which is independent of $V$. From now on,
we focus on the difference of the two differential entropy terms
in (\ref{cost_function}). Next, we note that using
Lemma~\ref{gradient_fisher_conditional}, we have
\begin{align}
h(\tilde{\bby}|V)-h(\bbx|V)=\frac{1}{2}\int_{\bzero}^{\tilde{\bbsigma}_Y}\bbj(\bbx+\bbn|V)d\bbsigma_N
\label{de_bruin_implies}
\end{align}
where $\bbn$ is zero-mean Gaussian random vector with covariance
matrix $\bbsigma_N$ satisfying $\bzero\preceq \bbsigma_N$. Next,
we find upper and lower bounds for (\ref{de_bruin_implies}). We
note that Lemma~\ref{lemma_change_in_fisher} implies the following
upper bound for $\bbj(\bbx+\bbn|V)$
\begin{align}
\bbj(\bbx+\bbn|V)\preceq
\left[\bbj^{-1}(\bbx|V)+\bbsigma_N\right]^{-1}
\label{upper_bound_fisher}
\end{align}
Using (\ref{upper_bound_fisher}) in (\ref{de_bruin_implies}) in
conjunction with Lemma~\ref{Shamai_s_lemma}, we get
\begin{align}
h(\tilde{\bby}|V)-h(\bbx|V)\leq
\frac{1}{2}\log\frac{|\bbj^{-1}(\bbx|V)+\tilde{\bbsigma}_Y|}{|\bbj^{-1}(\bbx|V)|}\label{upper_bound_cost}
\end{align}
We note that due to Lemma~\ref{lemma_conditional_crb_vector}, we
have $\bbj(\bbx|V)\succeq \bbk_{X|V}^{-1}\succ \bzero$, i.e.,
(\ref{upper_bound_cost}) is well-defined. Similarly, using
Lemma~\ref{lemma_change_in_fisher}, we have
\begin{align}
\bbj^{-1}(\bbx+\tilde{\bbn}_Y|V)-\tilde{\bbsigma}_Y \succeq
\bbj^{-1}(\bbx+\bbn|V)-\bbsigma_N,\quad \bbsigma_N \preceq
\tilde{\bbsigma}_Y
\end{align}
which implies
\begin{align}
\bbj(\bbx+\bbn|V)\succeq
\left[\bbj^{-1}(\bbx+\tilde{\bbn}_Y|V)-\tilde{\bbsigma}_Y+\bbsigma_N\right]^{-1}
\label{lower_bound_fisher}
\end{align}
Using (\ref{lower_bound_fisher}) in (\ref{de_bruin_implies}) in
conjunction with Lemma~\ref{Shamai_s_lemma}, we get
\begin{align}
h(\tilde{\bby}|V)-h(\bbx|V)\geq \frac{1}{2}
\log\frac{|\bbj^{-1}(\bbx+\tilde{\bbn}_Y|V)|}{|\bbj^{-1}(\bbx+\tilde{\bbn}_Y|V)-\tilde{\bbsigma}_Y|}
\label{lower_bound_cost}
\end{align}
Now, we rewrite the bounds in (\ref{upper_bound_cost}) and
(\ref{lower_bound_cost}). To this end, we define the following
function
\begin{align}
f(t)=\frac{1}{2}
\log\frac{|\bbk(t)+\tilde{\bbsigma}_Y|}{|\bbk(t)|},\quad 0\leq
t\leq 1\label{function}
\end{align}
where the matrix $\bbk(t)$ is given as follows
\begin{align}
\bbk(t)=t\bbj^{-1}(\bbx|V)+(1-t)\left[\bbj^{-1}(\bbx+\tilde{\bbn}_Y|V)-\tilde{\bbsigma}_Y\right]
\end{align}
Hence, using $f(t)$ in (\ref{function}), the bounds in
(\ref{upper_bound_cost}) and (\ref{lower_bound_cost}) can be
rewritten as follows:
\begin{align}
f(0)\leq h(\tilde{\bby}|V)-h(\bbx|V) \leq f(1)
\end{align}
Since $f(t)$ is continuous in $t$, there exists $t^*$ such that
\begin{align}
f(t^*)&= h(\tilde{\bby}|V)-h(\bbx|V)\\
&=\frac{1}{2} \log
\frac{|\bbk(t^*)+\tilde{\bbsigma}_Y|}{|\bbk(t^*)|} \label{fixing}
\end{align}
where $\bbk(t^*)$ is bounded as follows
\begin{align}
\bbj^{-1}(\bbx|V)\preceq \bbk(t^*)&\preceq
\bbj^{-1}(\bbx+\tilde{\bbn}_Y|V)-\tilde{\bbsigma}_Y\\
&\preceq \bbj^{-1}(\bbx+\bbn_Y|V)-\bbsigma_Y
\label{bound_the_fixed_one}
\end{align}
where we used the fact that $0\leq t^*\leq 1$ and
Lemma~\ref{lemma_change_in_fisher}. Thus, (\ref{fixing}) implies
that if we pick a Gaussian $\bar{V}$ satisfying
$\bbk_{X|\bar{V}}=\bbk(t^*)$, it provides the same value for the
cost function of $\bar{L}$ as the original $V$ does.

Next, we check whether this Gaussian $\bar{V}$ is feasible, i.e.,
whether it satisfies $\bbk_{X|\bar{V}Y}\preceq \bbd$. To this end,
using Lemma~\ref{lemma_connection}, we get
\begin{align}
\bbk_{X|\bar{V}Y}=\bbsigma_Y-\bbsigma_Y\bbj(\bby|\bar{V})\bbsigma_Y\label{connection_implies}
\end{align}
Since $\bar{V}$ is Gaussian,
Lemma~\ref{lemma_conditional_crb_vector} implies that
\begin{align}
\bbj(\bby|\bar{V})&=\bbk_{Y|\bar{V}}^{-1}\\
&=(\bbk_{X|\bar{V}}+\bbsigma_Y)^{-1} \label{markov_chain}
\end{align}
where (\ref{markov_chain}) follows from the fact that
$(\bar{V},\bbx)$ and $\bbn_Y$ are independent. Moreover, due to
(\ref{bound_the_fixed_one}), we have $\bbk_{X|\bar{V}}\preceq
\bbj^{-1}(\bby|V)-\bbsigma_Y$, which together with
(\ref{markov_chain}) imply the following
\begin{align}
\bbj(\bby|\bar{V})\succeq \bbj(\bby|V) \label{bound}
\end{align}
Using (\ref{bound}) in (\ref{connection_implies}), we get
\begin{align}
\bbk_{X|\bar{V}Y}&\preceq
\bbsigma_Y-\bbsigma_Y\bbj(\bby|V)\bbsigma_Y\\
&=\bbk_{X|VY}\label{connection_implies_1}\\
&\preceq \bbd \label{assumption}
\end{align}
where (\ref{connection_implies_1}) follows from
Lemma~\ref{lemma_connection} and (\ref{assumption}) is due to the
assumption that $V$ is feasible, i.e., $\bbk_{X|VY}\preceq \bbd$.
Equation (\ref{assumption}) implies that the constructed Gaussian
random vector $\bar{V}$ is feasible, i.e., for each feasible $V$,
there exists a feasible Gaussian $\bar{V}$ which provides the same
value for the cost function of $\bar{L}$; completing the first
step of the proof.

Hence, in view of this first step of the proof, we can restrict
$V$ to be Gaussian which leads to the following form for
$\bar{L}$:
\begin{align}
\bar{L}&=\min_{\substack{V\rightarrow \bbx \rightarrow
\tilde{\bby}\rightarrow \bby,\bbz\\ V~{\rm
is~Gaussian}\\\bbk_{X|VY}\preceq
\bbd}}~~I(V;\bbx)-I(V;\tilde{\bby})+I(\bbx;\bbz)\\
&= \min_{\substack{V\rightarrow \bbx \rightarrow
\tilde{\bby}\rightarrow \bby,\bbz\\ V~{\rm
is~Gaussian}\\\bbk_{X|V}\preceq
\bbf(\bbd)}}~~I(V;\bbx)-I(V;\tilde{\bby})+I(\bbx;\bbz)\label{Some_lemma_implies}\\
&=\min_{\bbk_{X|V}\preceq \bbf(\bbd)}~~
\frac{1}{2}\log\frac{|\bbk_X|}{|\bbk_{X|V}|}-\frac{1}{2}\log\frac{|\bbk_X+\tilde{\bbsigma}_Y|}{|\bbk_{X|V}+\tilde{\bbsigma}_Y|}
+\frac{1}{2}\log\frac{|\bbk_X+\bbsigma_Z|}{|\bbsigma_Z|}\\
&=
\frac{1}{2}\log\frac{|\bbk_X|}{|\bbf(\bbd)|}-\frac{1}{2}\log\frac{|\bbk_X+\tilde{\bbsigma}_Y|}{|\bbf(\bbd)+\tilde{\bbsigma}_Y|}
+\frac{1}{2}\log\frac{|\bbk_X+\bbsigma_Z|}{|\bbsigma_Z|}\label{final_step}
\end{align}
where (\ref{Some_lemma_implies}) follows from
Lemma~\ref{lemma_equivalence_1}, and (\ref{final_step}) comes from
the fact that
\begin{align}
\frac{|\bbk_{X|V}+\tilde{\bbsigma}_Y|}{|\bbk_{X|V}|}
\end{align}
is monotonically decreasing in the positive semi-definite matrices
$\bbk_{X|V}$; completing the proof of
Lemma~\ref{lemma_lower_bound}.

\section{Proof of Lemma~\ref{lemma_inclusion}}
\label{proof_of_lemma_inclusion}

We note that due to Theorem~\ref{theorem_Gauss_single}, we already
have single-letter descriptions for the regions
$\mathcal{R}_o(\bbd)$ and $\mathcal{R}_\alpha (\bbd)$. Thus, to
prove Lemma~\ref{lemma_inclusion}, it suffices to show that for
any given feasible $(U,V)$, these two regions satisfy the
relationship given in Lemma~\ref{lemma_inclusion}. We first note
the following Markov chains
\begin{align}
& U\rightarrow V \rightarrow \bbx \rightarrow
\bar{\bar{\bby}}_\alpha \rightarrow \bby \label{MC_y}\\
& U\rightarrow V \rightarrow \bbx \rightarrow
\bar{\bar{\bbz}}_\alpha \rightarrow \bbz \label{MC_z}
\end{align}
Next, we show that any feasible $(U,V)$ for the region
$\mathcal{R}_o(\bbd)$ is also feasible for the region
$\lim_{\alpha\rightarrow 0}\mathcal{R}_\alpha(\bbd)$. To this end,
we note that
\begin{align}
\bbd &\succeq \bbk_{X|VY} \\
&\succeq \bbk_{X|VY\bar{\bar{Y}}_\alpha} \label{cond_reduces_mmse_1}\\
&= \bbk_{X|V\bar{\bar{Y}}_\alpha} \label{MC_y_implies}
\end{align}
where (\ref{cond_reduces_mmse_1}) is due to the fact that
conditioning reduces MMSE and (\ref{MC_y_implies}) follows from
the Markov chain in (\ref{MC_y}). Moreover, it can be shown that
$\lim_{\alpha \rightarrow 0}\bbk_{X|V\bar{\bar{Y}}_\alpha} $
exists and is equal to $\bbk_{X|VY}$. Hence, this observation and
(\ref{MC_y_implies}) imply that $(U,V)$ is also feasible for the
region $\lim_{\alpha \rightarrow 0}\mathcal{R}_\alpha (\bbd)$.

Next, we show that for a given $(U,V)$, any rate inside the region
$\mathcal{R}_o(\bbd)$ is also inside $\lim_{\alpha \rightarrow
0}\mathcal{R}_\alpha (\bbd)$. To this end, for a given $(U,V)$, we
denote the minimum achievable rates in $\mathcal{R}_o(\bbd)$ and
$\mathcal{R}_\alpha (\bbd)$ by $R_o$ and $R_\alpha$, respectively.
Due to Theorem~\ref{theorem_Gauss_single}, we have
\begin{align}
R_o-R_\alpha &=
[I(V;\bbx)-I(V;\bby)]-[I(V;\bbx)-I(V;\bar{\bar{\bby}}_\alpha)] \\
&=I(V;\bar{\bar{\bby}}_\alpha)-I(V;\bby)\\
&=I(V;\bar{\bar{\bby}}_\alpha|\bby)\label{MC_y_implies_1}\\
&\geq 0
\end{align}
where (\ref{MC_y_implies_1}) comes from the Markov chain in
(\ref{MC_y}). Equation (\ref{MC_y_implies_1}) implies that any
achievable rate within the region $\mathcal{R}_o(\bbd)$ is also
included in the region $\lim_{\alpha \rightarrow
0}\mathcal{R}_\alpha (\bbd)$.

Finally, we show that for a given $(U,V)$, any achievable
information leakage inside the region $\mathcal{R}_o(\bbd)$ is
also inside $\lim_{\alpha \rightarrow 0}\mathcal{R}_\alpha
(\bbd)$. To this end, for a given $(U,V)$, we denote the minimum
information leakage in $\mathcal{R}_o(\bbd)$ and
$\mathcal{R}_\alpha (\bbd)$ by $I_{e,o}$ and $I_{e,\alpha}$,
respectively. Due to Theorem~\ref{theorem_Gauss_single}, we have
\begin{align}
I_{e,o}-I_{e,\alpha}&=\left[I(V;\bbx)-I(V;\bby|U)+I(\bbx;\bbz|U)\right]\nonumber \\
&\quad -
\big[I(V;\bbx)-I(V;\bar{\bar{\bby}}_\alpha|U)+I(\bbx;\bar{\bar{\bbz}}_\alpha|U)\big]\\
&=\big[I(V;\bar{\bar{\bby}}_\alpha|U)-I(V;\bby|U)\big] +\big
[I(\bbx;\bbz|U)-I(\bbx;\bar{\bar{\bbz}}_\alpha|U)\big]\\
&=I(V;\bar{\bar{\bby}}_\alpha|U,\bby) +\big
[I(\bbx;\bbz|U)-I(\bbx;\bar{\bar{\bbz}}_\alpha|U)\big]
\label{MC_y_implies_2} \\
&\geq I(\bbx;\bbz|U)-I(\bbx;\bar{\bar{\bbz}}_\alpha|U)\\
&\geq I(\bbx;\bbz)-I(\bbx;\bar{\bar{\bbz}}_\alpha)
\label{MC_z_implies} \\
&=\frac{1}{2} \log|\bbh_Z\bbk_X\bbh_Z^\top+\bbi|-\frac{1}{2}
\log\frac{|\bbk_X+\bbr_Z (\bblambda_Z+\alpha
\bbi)^{-2}\bbr_Z^\top|}{|\bbr_Z (\bblambda_Z+\alpha
\bbi)^{-2}\bbr_Z^\top|}\\
&=\frac{1}{2} \log|\bbh_Z\bbk_X\bbh_Z^\top+\bbi|-\frac{1}{2}
\log\frac{|\bbk_X+\bbr_Z (\bblambda_Z+\alpha \bbi)^{-1}\bbq_Z^\top
\bbq_Z(\bblambda_Z+\alpha \bbi)^{-1}\bbr_Z^\top|}{|\bbr_Z
(\bblambda_Z+\alpha \bbi)^{-1}\bbq_Z^\top
\bbq_Z(\bblambda_Z+\alpha \bbi)^{-1}\bbr_Z^\top|}\\
&=\frac{1}{2} \log|\bbh_Z\bbk_X\bbh_Z^\top+\bbi|-\frac{1}{2} \log
|\bbq_Z (\bblambda_Z+\alpha \bbi)\bbr_Z^\top \bbk_X\bbr_Z
(\bblambda_Z+\alpha \bbi)\bbq_Z^\top+\bbi| \label{inclusion}
\end{align}
where (\ref{MC_y_implies_2}) comes from the Markov chain in
(\ref{MC_y}) and (\ref{MC_z_implies}) follows from the Markov
chain in (\ref{MC_z}). Equation (\ref{inclusion}) implies that
\begin{align}
\lim_{\alpha \rightarrow 0} I_{e,o}-I_{e,\alpha}& \geq \frac{1}{2}
\log|\bbh_Z\bbk_X\bbh_Z^\top+\bbi|-\lim_{\alpha\rightarrow 0
}\frac{1}{2} \log |\bbq_Z (\bblambda_Z+\alpha \bbi)\bbr_Z^\top
\bbk_X\bbr_Z (\bblambda_Z+\alpha \bbi)\bbq_Z^\top+\bbi|\\
&= \frac{1}{2} \log|\bbh_Z\bbk_X\bbh_Z^\top+\bbi|-\frac{1}{2} \log
|\bbq_Z \bblambda_Z\bbr_Z^\top \bbk_X\bbr_Z
\bblambda_Z\bbq_Z^\top+\bbi| \label{continuity}\\
&= \frac{1}{2} \log|\bbh_Z\bbk_X\bbh_Z^\top+\bbi|-\frac{1}{2}
\log|\bbh_Z\bbk_X\bbh_Z^\top+\bbi|\\
&=0 \label{inclusion_final}
\end{align}
where (\ref{continuity}) comes from the continuity of the
determinant in positive semi-definite matrices. Equation
(\ref{inclusion_final}) implies that any achievable information
leakage in the region $\mathcal{R}_o(\bbd)$ is also inside the
region $\lim_{\alpha\rightarrow 0}\mathcal{R}_\alpha(\bbd)$;
completing the proof of Lemma~\ref{lemma_inclusion}.

\section{Proof of Theorem~\ref{theorem_outer_general}}

\label{proof_of_theorem_outer_general}

We start the proof of Theorem~\ref{theorem_outer_general} by first
expressing Theorem~\ref{theorem_outer} for the side information
model given by
(\ref{general_alpha_1_aligned})-(\ref{general_alpha_2_aligned}).
In other words, we first provide an outer bound for the region
$\mathcal{R}_\alpha (\bbd)$ by using Theorem~\ref{theorem_outer}.
To this end, to be able to use Theorem~\ref{theorem_outer}, we
need $\bbd \preceq \bbk_{X|\bar{\bar{Y}}_\alpha}$. However, since
we originally have $\bbd \preceq \bbk_{X|Y}$ and
$\bbk_{X|\bar{\bar{Y}}_\alpha}\preceq \bbk_{X|Y}$, where the
latter one follows from the Markov chain $\bbx\rightarrow
\bar{{\bby}}_\alpha\rightarrow \bby$ and the fact that
conditioning reduces MMSE, $\bbk_{X|\bar{\bar{Y}}_\alpha}-\bbd$
might be indefinite. However, the only place we use the condition
$\bbd\preceq \bbk_{X|Y}$ is to be able to show the equivalence
between $\bbk_{X|VY}\preceq \bbd$ and $\bbk_{X|V}\preceq
\bbf(\bbd)$ for Gaussian $V$ in Lemma~\ref{lemma_equivalence_1}.
In particular, we only need the fact that $\bbsigma_Y-\bbd$ is
non-singular to show this equivalence, and which is implied by
$\bbd\preceq \bbk_{X|Y}$. However, still there might be distortion
matrices $\bbd$ for which although we have non-singular
$\bbsigma_Y-\bbd$, the condition $\bbd\preceq \bbk_{X|Y}$ is not
satisfied. Hence, if we can find an $\alpha^* $ such that
\begin{align}
\bbsigma_{Y,\alpha}-\bbd \succ \bzero,\quad 0<\alpha\leq \alpha^*
\label{desired_condition}
\end{align}
we can still use Theorem~\ref{theorem_outer} to obtain an outer
bound for the region $\mathcal{R}_\alpha (\bbd)$. Now, we
establish the existence of such an $\alpha^*$. Using the
assumption $\bbd\preceq \bbk_{X|Y}$, we have
\begin{align}
\bbd\preceq \bbk_{X|Y}=(\bbk_X^{-1}+\bbh_Y^\top\bbh_Y)^{-1}
\label{non_singular}
\end{align}
where the equality follows from (\ref{future_use_x}). Equation
(\ref{non_singular}) implies that
\begin{align}
\bzero &\prec \bbd^{-1}-\bbh_Y^\top\bbh_Y \label{future_use_y}\\
&= \bbd^{-1}-\bbr_Y\bblambda_Y^2\bbr_Y^\top
\end{align}
where we use the singular value decomposition of $\bbh_Y$. Thus,
since $\bbd^{-1}-\bbr_Y\bblambda_Y^2\bbr_Y^\top$ is strictly
positive definite, there exists $0< \beta$ such that
\begin{align}
 \bbd^{-1}-\bbr_Y\bblambda_Y^2\bbr_Y^\top &\succ
\beta^2\bbi\\
&=\beta^2 \bbr_Y\bbr_Y^\top
\end{align}
which implies
\begin{align}
 \bbd^{-1}\succ \bbr_Y(\bblambda_Y^2+\beta^2)\bbr_Y^\top
\end{align}
which, in turn, implies the existence of an $\alpha^*$ such that
\begin{align}
 \bbd^{-1}\succ \bbr_Y(\bblambda_Y+\alpha)^2\bbr_Y^\top,\quad
0<\alpha \leq \alpha^* \label{extension}
\end{align}
Hence, using the definition of $\bbsigma_{Y,\alpha}$ in
(\ref{extension}), we get
\begin{align}
\bbd^{-1}\succ \bbsigma_{Y,\alpha}^{-1},\quad  0<\alpha \leq
\alpha^*
\end{align}
which is equivalent to the desired condition in
(\ref{desired_condition}) which is needed to use
Theorem~\ref{theorem_outer} to obtain an outer bound for the
region $\mathcal{R}_\alpha(\bbd)$. Hence, assuming that
$0<\alpha\leq\alpha^*$, an outer bound for the region
$\mathcal{R}_\alpha(\bbd)$ can be written as the union of rate and
information leakage $(R,I_e)$ pairs satisfying
\begin{align}
R&\geq \frac{1}{2} \log
\frac{|\bbk_{X|\bar{\bar{Y}}_\alpha}|}{|\bbd|}=\frac{1}{2}\log\frac{|\bbk_X|}{|\bbf_\alpha(\bbd)|}-\frac{1}{2}\log\frac{|\bbk_X+\bbsigma_{Y,\alpha}|}{|\bbf_\alpha(\bbd)+\bbsigma_{Y,\alpha}|}
\label{rate_bound_alpha}\\
I_e & \geq~\min_{\substack{0\preceq \bbk_{X|V}\preceq
\bbk_{X|U}\preceq \bbk_X\\ \bbk_{X|V}\preceq \bbf_\alpha(\bbd)}}
\frac{1}{2}\log\frac{|\bbk_X|}{|\bbk_{X|V}|}-\frac{1}{2}\log\frac{|\bbk_{X|U}+\bbsigma_{Y,\alpha}|}{|\bbk_{X|V}+\bbsigma_{Y,\alpha}|}
+\frac{1}{2}
\log\frac{|\bbk_{X|U}+\bbsigma_{Z,\alpha}|}{|\bbsigma_{Z,\alpha}|}
\label{leakage_bound_alpha}
\end{align}
where
$\bbf_\alpha(\bbd)=\bbsigma_{Y,\alpha}(\bbsigma_{Y,\alpha}-\bbd)^{-1}\bbsigma_{Y,\alpha}-\bbsigma_{Y,\alpha}$.
We now find the limiting region that comes from the one described
by (\ref{rate_bound_alpha})-(\ref{leakage_bound_alpha}) as
$\alpha\rightarrow 0$. To this end,we introduce the following
lemma that will be used subsequently.
\begin{Lem}
\label{lemma_matrix_limits}
\begin{align}
\lim_{\alpha\rightarrow0}
\bbk_{X|\bar{\bar{Y}}_\alpha}&=\bbk_{X|Y}\label{matrix_limit_1}\\
\lim_{\alpha\rightarrow
0}\bbf_\alpha(\bbd)&=(\bbd^{-1}-\bbh_{Y}^\top\bbh_Y)^{-1}
\label{matrix_limit_2}
\end{align}
\end{Lem}
The proof of Lemma~\ref{lemma_matrix_limits} is given in
Appendix~\ref{proof_of_lemma_matrix_limits}.

We first consider the rate bound in (\ref{rate_bound_alpha}) as
follows
\begin{align}
\lim_{\alpha\rightarrow 0}\frac{1}{2} \log
\frac{|\bbk_{X|\bar{\bar{Y}}_\alpha}|}{|\bbd|}=\frac{1}{2} \log
\frac{|\bbk_{X|Y}|}{|\bbd|} \label{limit_rate_bound}
\end{align}
which follows from the continuity of the determinant in positive
semi-definite matrices and (\ref{matrix_limit_1}). Similarly, for
the second expression in the rate bound in
(\ref{rate_bound_alpha}), we have
\begin{align}
\lefteqn{\lim_{\alpha\rightarrow
0}\frac{1}{2}\log\frac{|\bbk_X|}{|\bbf_\alpha(\bbd)|}-\frac{1}{2}\log\frac{|\bbk_X+\bbsigma_{Y,\alpha}|}{|\bbf_\alpha(\bbd)+\bbsigma_{Y,\alpha}|}}\nonumber\\
&=\frac{1}{2}\log\frac{|\bbk_X|}{|(\bbd^{-1}-\bbh_Y^\top\bbh_Y)^{-1}|}
-\lim_{\alpha \rightarrow 0}
\frac{1}{2}\log\frac{|\bbk_X+\bbsigma_{Y,\alpha}|}{|\bbf_\alpha(\bbd)+\bbsigma_{Y,\alpha}|}\label{matrix_limit_2_implies_1}\\
&=\frac{1}{2}\log\frac{|\bbk_X|}{|(\bbd^{-1}-\bbh_Y^\top\bbh_Y)^{-1}|}-\lim_{\alpha
\rightarrow 0}
\frac{1}{2}\log\frac{|\bbk_X+\bbr_Y(\bblambda_Y+\alpha\bbi)^{-2}\bbr_Y^\top|}{|\bbf_\alpha(\bbd)+\bbr_Y(\bblambda_Y+\alpha\bbi)^{-2}\bbr_Y^\top|}\label{def_noise_y_alpha_implies_1}
\\
&=\frac{1}{2}\log\frac{|\bbk_X|}{|(\bbd^{-1}-\bbh_Y^\top\bbh_Y)^{-1}|}-\lim_{\alpha
\rightarrow 0}
\frac{1}{2}\log\frac{|\bbk_X+\bbr_Y(\bblambda_Y+\alpha\bbi)^{-1}\bbq_Y^\top\bbq_Y(\bblambda_Y+\alpha\bbi)^{-1}\bbr_Y^\top|}{|\bbf_\alpha(\bbd)+\bbr_Y(\bblambda_Y+\alpha\bbi)^{-1}\bbq_Y^\top\bbq_Y(\bblambda_Y+\alpha\bbi)^{-1}\bbr_Y^\top|}
\\
&=\frac{1}{2}\log\frac{|\bbk_X|}{|(\bbd^{-1}-\bbh_Y^\top\bbh_Y)^{-1}|}-\lim_{\alpha
\rightarrow 0} \frac{1}{2}\log\frac{|\bbq_Y(\bblambda_Y+\alpha
\bbi)\bbr_Y^\top\bbk_X\bbr_Y(\bblambda_Y+\alpha\bbi)\bbq_Y^\top+\bbi|}{|\bbq_Y(\bblambda_Y+\alpha
\bbi)\bbr_Y^\top\bbf_\alpha(\bbd)\bbr_Y(\bblambda_Y+\alpha\bbi)\bbq_Y^\top+\bbi|}\\
&=\frac{1}{2}\log\frac{|\bbk_X|}{|(\bbd^{-1}-\bbh_Y^\top\bbh_Y)^{-1}|}-
\frac{1}{2}\log\frac{|\bbq_Y\bblambda_Y\bbr_Y^\top\bbk_X\bbr_Y\bblambda_Y\bbq_Y^\top+\bbi|}{|\bbq_Y\bblambda_Y
\bbr_Y^\top(\bbd^{-1}-\bbh_Y^\top\bbh_Y)^{-1}\bbr_Y\bblambda_Y\bbq_Y^\top+\bbi|}\label{matrix_limit_2_implies_2}\\
&=\frac{1}{2}\log\frac{|\bbk_X|}{|(\bbd^{-1}-\bbh_Y^\top\bbh_Y)^{-1}|}-
\frac{1}{2}\log\frac{|\bbh_Y\bbk_X\bbh_Y^\top+\bbi|}{|\bbh_Y(\bbd^{-1}-\bbh_Y^\top\bbh_Y)^{-1}\bbh_Y^\top+\bbi|}
\label{svd_implies_2}
\end{align}
where (\ref{matrix_limit_2_implies_1}) is due to the continuity of
the determinant in positive semi-definite matrices and
(\ref{matrix_limit_2}), (\ref{def_noise_y_alpha_implies_1}) comes
from the definition of $\bbsigma_{Y,\alpha}$,
(\ref{matrix_limit_2_implies_2}) comes from the continuity of the
determinant in positive semi-definite matrices and
(\ref{matrix_limit_2}), and (\ref{svd_implies_2}) is obtained by
using the singular value decomposition of $\bbh_Y$. Hence,
(\ref{limit_rate_bound}) and (\ref{svd_implies_2}) imply that any
rate $R$ inside the region $\lim_{\alpha\rightarrow
0}\mathcal{R}_\alpha (\bbd)$ satisfies
\begin{align}
R&\geq\frac{1}{2} \log \frac{|\bbk_{X|Y}|}{|\bbd|}\\
&=\frac{1}{2}\log\frac{|\bbk_X|}{|(\bbd^{-1}-\bbh_Y^\top\bbh_Y)^{-1}|}-
\frac{1}{2}\log\frac{|\bbh_Y\bbk_X\bbh_Y^\top+\bbi|}{|\bbh_Y(\bbd^{-1}-\bbh_Y^\top\bbh_Y)^{-1}\bbh_Y^\top+\bbi|}
\label{rate_bound_alpha_limit}
\end{align}
Following a similar analysis, the limit of the information leakage
in (\ref{leakage_bound_alpha}) can be found as
\begin{align}
\min_{\substack{\bzero \preceq \bbk_{X|V}\preceq \bbk_{X|U}\preceq
\bbk_X\\ \bbk_{X|V}\preceq (\bbd^{-1}-\bbh_Y^\top\bbh_Y)^{-1}}}
\frac{1}{2}\log\frac{|\bbk_X|}{|\bbk_{X|V}|}-\frac{1}{2}\log\frac{|\bbh_Y\bbk_{X|U}\bbh_Y^\top+\bbi|}{|\bbh_Y\bbk_{X|V}\bbh_Y^\top+\bbi|}
+\frac{1}{2} \log|\bbh_Y\bbk_{X|U}\bbh_Y^\top+\bbi|
\label{leakage_bound_alpha_limit}
\end{align}
which implies that any information leakage $I_e$ inside the region
$\lim_{\alpha\rightarrow 0}\mathcal{R}_\alpha(\bbd)$ should be
larger than (\ref{leakage_bound_alpha_limit}); completing the
proof of Theorem~\ref{theorem_outer_general}.

\section{Proof of Lemma~\ref{lemma_matrix_limits}}
\label{proof_of_lemma_matrix_limits}

We first prove the following lemma which will be used
subsequently.
\begin{Lem}
\label{lemma_matrix_limit_core} Let
$\bbk(\alpha)=(\bba+f(\alpha)\bbb)^{-1},~0< \alpha \leq\alpha^*$,
where $\bba\succ f(\alpha)\bbb\succeq \bzero,~0\leq \alpha \leq
\alpha^*$ and $f(\alpha)$ is continuous in $\alpha$. Then, we have
\begin{align}
\lim_{\alpha\rightarrow 0}\bbk(\alpha)=(\bba+f(0)\bbb)^{-1}
\end{align}
\end{Lem}
\begin{proof}
In the proof of this lemma, we use the fact that if
$\lim_{n\rightarrow \infty}\bbc^n=\bzero$, we have
\begin{align}
(\bbi+\bbc)^{-1}=\sum_{n=0}^{\infty}(-1)^n\bbc^n \label{the_sum}
\end{align}
where $\bbc^0=\bbi$~\cite[page 19]{matrix_cookbook}. Now, we
consider
\begin{align}
\bbk(\alpha)&= (\bba+f(\alpha)\bbb)^{-1}\\
&=\bba^{-1/2}(\bbi+f(\alpha)\bba^{-1/2}\bbb\bba^{-1/2})^{-1}\bba^{-1/2}
\label{before_sum_applies}
\end{align}
where due to $\bba\succ f(\alpha)\bbb\succeq \bzero$, we have
$\bbi\succ  f(\alpha)\bba^{-1/2}\bbb\bba^{-1/2} \succeq \bzero$
which implies
\begin{align}
\lim_{n\rightarrow
\infty}\left(f(\alpha)\bba^{-1/2}\bbb\bba^{-1/2}\right)^n=\bzero
\end{align}
Hence, we can use (\ref{the_sum}) in (\ref{before_sum_applies}) to
get
\begin{align}
\bbk(\alpha)
&=\bba^{-1/2}\left[\sum_{n=0}^{\infty}(-1)^nf^n(\alpha)(\bba^{-1/2}\bbb\bba^{-1/2})^{n}\right]\bba^{-1/2}
\end{align}
which implies
\begin{align}
\lim_{\alpha\rightarrow 0}\bbk(\alpha) &=\lim_{\alpha\rightarrow
0}\bba^{-1/2}\left[\sum_{n=0}^{\infty}(-1)^nf^n(\alpha)(\bba^{-1/2}\bbb\bba^{-1/2})^{n}\right]\bba^{-1/2}\\
&=\bba^{-1/2}\left[\sum_{n=0}^{\infty}(-1)^nf^n(0)(\bba^{-1/2}\bbb\bba^{-1/2})^{n}\right]\bba^{-1/2}\\
&=\bba^{-1/2}\left[\bbi+f(0)\bba^{-1/2}\bbb\bba^{-1/2}\right]^{-1}\bba^{-1/2}
\label{after_sum_applies}\\
&=(\bba+f(0)\bbb)^{-1}
\end{align}
where (\ref{after_sum_applies}) comes from (\ref{the_sum});
completing the proof of Lemma~\ref{lemma_matrix_limit_core}.
\end{proof}

We now consider (\ref{matrix_limit_1}) in
Lemma~\ref{lemma_matrix_limits} as follows
\begin{align}
\bbk_{X|\bar{\bar{\bby}}_\alpha}&=\bbk_X(\bbk_X+\bbsigma_{Y,\alpha})^{-1}\bbsigma_{Y,\alpha}
\label{future_use_implies}\\
&=(\bbk_{X}^{-1}+\bbsigma_{Y,\alpha}^{-1})^{-1} \\
&=\left[\bbk_{X}^{-1}+\bbr_{Y}(\bblambda_Y+\alpha\bbi)^{2}\bbr_Y^\top\right]^{-1}
\label{def_noise_y_alpha} \\
&=\left[\bbk_{X}^{-1}+\bbr_{Y}\bblambda_Y^{2}\bbr_Y^\top+\bbr_{Y}(2\alpha
\bblambda_Y+\alpha^2\bbi)\bbr_Y^\top\right]^{-1}
\label{future_use_again}
\end{align}
where $0<\alpha\leq \alpha^*$. Equation (\ref{future_use_implies})
comes from (\ref{future_use}), (\ref{def_noise_y_alpha}) is due to
the definition of $\bbsigma_{Y,\alpha}$. We note that
$\bbk_{X}^{-1}+\bbr_{Y}\bblambda_Y^{2}\bbr_Y^\top \succ \bzero$,
and thus, $\alpha^*$ can be selected to ensure that
\begin{align}
\bbk_{X}^{-1}+\bbr_{Y}\bblambda_Y^{2}\bbr_Y^\top\succ
\bbr_{Y}(2\alpha \bblambda_Y+\alpha^2\bbi)\bbr_Y^\top
\end{align}
for all $0\leq \alpha\leq \alpha^*$. Hence, we can use Lemma
\ref{lemma_matrix_limit_core} in (\ref{future_use_again}) to get
\begin{align}
\lim_{\alpha \rightarrow 0}\bbk_{X|\bar{\bar{\bby}}_\alpha}&=
\left[\bbk_{X}^{-1}+\bbr_{Y}\bblambda_Y^{2}\bbr_Y^\top\right]^{-1}\\
&=\left[\bbk_{X}^{-1}+\bbr_{Y}\bblambda_Y\bbq_Y^\top\bbq_Y\bblambda_Y
\bbr_Y^\top\right]^{-1}\\
&=(\bbk_X^{-1}+\bbh_Y^T\bbh_Y)^{-1}\label{svd_implies}\\
&=\bbk_{X|Y} \label{future_use_x_implies}
\end{align}
where (\ref{svd_implies}) comes from the singular value
decomposition of $\bbh_Y$ and (\ref{future_use_x_implies}) is due
to (\ref{future_use_x}); completing the proof of
(\ref{matrix_limit_1}).

Next, we consider (\ref{matrix_limit_2}) in
Lemma~\ref{lemma_matrix_limits} as follows
\begin{align}
\bbf_\alpha (\bbd)&=
\bbsigma_{Y,\alpha}(\bbsigma_{Y,\alpha}-\bbd)^{-1}\bbsigma_{Y,\alpha}-\bbsigma_{Y,\alpha}
\\
&=\bbsigma_{Y,\alpha}(\bbsigma_{Y,\alpha}-\bbd)^{-1}\bbd\\
&=(\bbd^{-1}-\bbsigma_{Y,\alpha}^{-1})^{-1}\\
&=(\bbd^{-1}-\bbr_Y(\bblambda_Y+\alpha\bbi)^{2}\bbr_Y^\top)^{-1}\label{def_noise_y_alpha_implies}
\\
&=\left[\bbd^{-1}-\bbr_Y\bblambda_Y^2\bbr_{Y}^\top-\bbr_Y(2\alpha\bblambda_Y+\alpha^2\bbi)\bbr_Y^\top\right]^{-1}\\
&=\left[\bbd^{-1}-\bbr_Y\bblambda_Y\bbq_Y^\top\bbq_Y\bblambda_Y
\bbr_{Y}^\top-\bbr_Y(2\alpha\bblambda_Y+\alpha^2\bbi)\bbr_Y^\top\right]^{-1}\\
&=\left[\bbd^{-1}-\bbh_Y^\top
\bbh_{Y}-\bbr_Y(2\alpha\bblambda_Y+\alpha^2\bbi)\bbr_Y^\top\right]^{-1}
\label{svd_implies_1}
\end{align}
where $0<\alpha\leq \alpha^*$. Equation
(\ref{def_noise_y_alpha_implies}) comes from the definition of
$\bbsigma_{Y,\alpha}$ and (\ref{svd_implies_1}) is obtained by
using the singular value decomposition of $\bbh_Y$. We note that
$\bbd^{-1}-\bbh_Y^\top\bbh_Y$ is strictly positive definite as
(\ref{future_use_y}) indicates, and hence, there exists an
$\alpha^*$ such that
\begin{align}
\bbd^{-1}-\bbh_Y^\top \bbh_{Y}\succ
\bbr_Y(2\alpha\bblambda_Y+\alpha^2\bbi)\bbr_Y^\top
\end{align}
for all $0\leq \alpha\leq \alpha^*$. Consequently, we can use
Lemma~\ref{lemma_matrix_limit_core} in (\ref{svd_implies_1}) to
get
\begin{align}
\lim_{\alpha \rightarrow
0}\bbf_\alpha(\bbd)=(\bbd^{-1}-\bbh_Y^\top\bbh_Y)^{-1}
\end{align}
which completes the proof of Lemma~\ref{lemma_matrix_limits}.

\bibliographystyle{unsrt}
\bibliography{IEEEabrv,references2}
\end{document}